\documentclass[12pt
]{iopart}

\usepackage{hyperref}
\usepackage{graphicx} % Required for inserting images
\newcommand{\newblock}{}

\usepackage[numbers]{natbib}
\expandafter\let\csname equation*\endcsname\relax

\expandafter\let\csname endequation*\endcsname\relax

\usepackage{amsmath, amssymb, amsthm}
\usepackage{braket}
\usepackage{nicefrac}
\usepackage{bbm}
\usepackage[english]{babel}
\usepackage{tikz}
\usepackage{todonotes}

\usepackage{pgfplots}
\pgfplotsset{compat=1.17}
\presetkeys%
    {todonotes}{inline,backgroundcolor=orange}{}

\newtheorem{theorem}{Theorem}[section]
\newtheorem{definition}[theorem]{Definition}
\newtheorem{lemma}[theorem]{Lemma}%lemma

\newtheorem{proposition}[theorem]{Proposition}
\newtheorem{remark}[theorem]{Remark}
\newtheorem{example}[theorem]{Example}

\newcommand{\SU}{\operatorname{SU}}
\newcommand{\su}{\mathfrak{su}}

\newcommand{\CNOT}{\operatorname{CNOT}}

\begin{document}

\title{Scaling of symmetry-restricted Lie groups}
%\author{Authors removed for peer review}
\author{Maximilian Balthasar Mansky$^1$, Miguel Armayor Martinez$^1$, Alejandro Bravo de la Serna$^1$, Santiago Londoño Castillo$^1$, Dimitra Nikolaidou$^1$, Gautham Sathish$^1$, Zhihao Wang$^1$, Sebastian Wölckert$^1$ and Claudia Linnhoff-Popien$^1$}
\address{$^1$Department of Informatics, LMU Munich, Oettingenstrasse 67, 80538 Munich,. Germany}
\ead{maximilian-balthasar.mansky@ifi.lmu.de}

\vspace{10pt}
\begin{indented}
\item[]February 2025
\end{indented}

\begin{abstract}

Symmetries are omnipresent in physics and have been used to reduce the number of degrees of freedom of systems. In this work, we investigate the properties of $\mathcal{M}\SU(N)$, $\mathcal{M}$-invariant subspaces of the special unitary Lie group $\SU(N)$. This group is relevant to quantum computing and quantum systems in general.  We demonstrate that for certain choices of $\mathcal{M}$, the subset $\mathcal{M}\SU(N)$ inherits many topological and group properties from $\SU(N)$. We then present a combinatorial method for computing the dimension of such subspaces when $\mathcal{M}$ is a representation of a permutation group acting on qubits $\left(G\SU(N)\right)$, or a Hamiltonian $\left(H^{(N)}\SU(N)\right)$. The Kronecker product of $\su(2)$ matrices is employed to construct the Lie algebras associated with different permutation-invariant groups $G\SU(N)$. Numerical results on the number of dimensions support the developed theory. %\cite{ragone2023unified}.
\end{abstract}
%\keywords{Lie groups, Symmetry}
\submitto{\jpa}

\maketitle
%\tableofcontents

%\listoftodos

\section{Introduction}

Invariant systems are common across many areas of science. They appear in the rotational symmetry of a classic mechanical problem \cite{landau_mechanics_1976}, position-free wave ansatzes in electro-dynamics \cite{landau_classical_1975} or the treatment of quantum-mechanical structures \cite{landau_quantum_1977}. The reduction of problem size as its number of dimensions allows for an easier treatment of the problem at hand, down to analytical solutions. 
Symmetries also appear in the context of questions in computer science, for example the properties of graphs. Some global properties of graphs are invariant under transformations, for example the graph connectedness property under a permutation of nodes.

The solution to a problem is a transformation from the problem to the solution is located somewhere in a computational space. It can be found via algorithms, deep thought, luck or machine learning. Machine learning provides a framework to reason about solution spaces and their sizes. (Classical) Machine learning works by taking a feature vector and transforming it with successive non-linear transformations. This structure has been proven to be able to approximate any continuous function, known as Cybenko's theorem \cite{cybenko_approximation_1989}. There is a plethora of different architectures suitable to different tasks \cite{zhang_new_2010, vaswani_attention_2017, dhillon_convolutional_2020}. In this work, we focus on a particular subfield, quantum machine learning \cite{biamonte_quantum_2017, cerezo_variational_2021}. 

Quantum computers form a new computational paradigm, where the information is no longer processed in discrete bits but rather represented on quantum systems, with distinctively quantum effects, such as superposition and entanglement, now part of the calculation. In a quantum computer, quantum circuits perform calculations and can be represented as a series of one- and two-qubit gates \cite{nielsen_quantum_2010}. In a variational quantum circuit, some of the gates contain parameters that can be adjusted to fit input to expected output, much like in a machine learning system \cite{biamonte_quantum_2017, cerezo_variational_2021}. Whether quantum machine learning provides an advantage over classical approaches remains a topic of debate \cite{schuld_effect_2021, bowles_better_2024}. Here we do not evaluate quantum over classical, but rather focus on the properties of quantum circuits and in particular cases where there is an external symmetry posed on them.

%The larger concept of a symmetry can be included into a computational ansatz. This is true for classical computing as well as for quantum computing. The latter is a new computing paradigm based on quantum mechanics to transform states that represent computational information into each other. Rather than letting the system evolve according to its physical interaction, the evolution is restricted with control elements called gates.

The restriction of the computational space to that of the problem can have a positive impact on the accuracy and training performance of the computation \cite{larocca_group-invariant_2022, meyer_exploiting_2023, schatzki_theoretical_2024, ragone_unified_2024, mansky_permutation-invariant_2023}. Related to the performance improvement is the size of available solution space, as defined as the number of parameters that can influence the result of the calculation \cite{ragone_unified_2024, cerezo_does_2023}. We provide a framework for calculating the number of parameters in the quantum circuit ansatz depending on the underlying problem symmetry. Construction of symmetry-invariant quantum circuits is straightforward in many cases. While not the focus of this work, we explain how to construct a quantum circuit invariant under any arbitrary discrete symmetry that can be expressed as a permutation. Construction methods for symmetries based on other structures exist \cite{nguyen_theory_2024}.

In this work, we examine the mathematical spaces related to these restrictions and prove that they are well-behaved, in particular that they are a connected Lie subgroup. The number of dimensions can be explicitly computed using combinatorics. Our results confirm and extend the results existing in literature \cite{schatzki_theoretical_2024, mansky_permutation-invariant_2023, barut_dynamical_1965, wiersema_classification_2023}. Quantum circuits can be explicitly constructed for permutation-based symmetries. For restrictions based on a set of matrices, the number of dimensions can be explicitly determined. We do not consider continuous symmetries, in contrast to the literature \cite{schatzki_theoretical_2024}, as the restricted space is not always well-behaved. In particular, the subgroup is no longer connected for high-dimensional continuous groups set as the restriction. We explore a numerical evaluation of  quantum circuits that are part of these symmetry-restricted Lie groups and find that their number of dimensions is proportional to the size of the symmetry-restricted space to which they relate. This indicates that the symmetry restriction is dominates the size of the quantum circuit. The behaviour is independent of the exact construction of the quantum circuit.

After a literature review in section \ref{sec:related-work}, we provide a mathematical framework for the treatment of symmetry invariant finite dimensional subgroups of the special unitary group $\SU$. The subgroups are well-behaved under weak assumptions. The growth of the number of dimensions of these subspaces is detailed in section \ref{sec:scaling}, with numerical results confirming the theoretical findings. The conclusion in section \ref{sec:discussion} closes our work and details the impact of our results to quantum computing.

\section{Related work}\label{sec:related-work}

Within the realm of physics, invariant subgroups are necessary for the treatment of restricted structures. The permutation invariance appears in the context of the Bose-Einstein condensate \cite{ziff_ideal_1977}. The requirement also has impact on the entanglement properties of the involved particles \cite{markham_entanglement_2011, yadin_thermodynamics_2023}. The properties of symmetry-restricted states has received some more attention than the corresponding operators \cite{toth_entanglement_2009, arrasmith_effect_2021, holmes_connecting_2022}. Although the two topics of states and operators are related, the results cannot be readily transferred, as the relevant mathematical spaces are different. 

For operators, the research on barren plateaus in quantum computing has explored the subalgebras and symmetry-restricted groups of $SU(2^n)$ to a limited extend \cite{larocca_group-invariant_2022, meyer_exploiting_2023, ragone_unified_2024, 
mansky_permutation-invariant_2023, nguyen_theory_2024,  marvian_restrictions_2022}. The treatment of the subgroups and subalgebras is accidental rather than methodical, with key concepts such as connectivity left untreated. For the treatment of barren plateaus as an inherent result of exponentially scaling subgroups it is sufficient, but it leaves a gap in the full treatment. We fill this gap for symmetry-based and Hamiltonian-based subgroups. These symmetry-restricted subspaces can also be found in classical machine learning, in theory \cite{bronstein_geometric_2021, klus_symmetric_2021} and application, for example in graph neural networks \cite{zhou_graph_2020}. The results are not immediately transferable, as the structure of the covered space is different, with quantum machine learning working on the closed compact $SU(2^n)$ and classical machine learning on matrices in $\mathbb{R}^{n\times m}$ with a complex geometry induced by the activation functions.

On the side of mathematics, the exploration of symmetry-restricted subgroups happens on an ad-hoc basis. The  subject receives some treatment in \cite{dungey_analysis_2013}. A  treatment of symmetric states as the counterpart to the operations considered here can be found in \cite{harrow_church_2013}.

\section{Symmetry Invariant Subgroups of $\SU(n)$}\label{sec:theoretical-framework}

    The framework of quantum computing relies on the same mathematics as that of quantum mechanics, where quantum states can be described as rays $|\phi\rangle$ on a complex vector space and operations between states as matrices in the special unitary group. We provide a self-contained introduction to some of the properties of $SU(n)$, in particular the $\SU(2^N)$ group relevant for quantum computing, for a more detailed account of the subject we refer the reader to \cite{nielsen_quantum_2010}.

We note that the terms invariance and equivariance are used somewhat interchangeable in the literature. In the following, we call objects invariant if they are unchanged under the symmetry operation – A subgroup is invariant under a symmetry operation. Individual elements within the subgroup are equivariant if they change, but are still elements of that subgroup.

\subsection{The Lie group $\SU(n)$ and its associated Lie algebra $\su(n)$}
\label{sec:propertiesSU}
    % Definition of SU(N)
    \begin{definition}
        The special unitary group $\SU(n)$ is a Lie group equipped with the matrix product, whose elements are $n\times n$ unitary matrices with determinant one.
    \end{definition}
    
    Topologically $\SU(n)$ has the structure of a smooth real manifold of (real) dimension $n^2-1$, which is compact and path connected, i.e. for any $A,\,B \in \SU{(N)}$ there exists a continuous function $f:[0,1]\rightarrow \SU(n)$ such that $f(0)=A$, $f(1)=B$, and it is thus a connected manifold.
\begin{definition}
    \label{def:su(N)}
        The real matrix Lie algebra $\su(n)$ is a real vector space of dimension $n^2-1$ with elements $n \times n$ matrices with complex entries which satisfy the properties of:
    \begin{enumerate}
        \item Skew hermiticity
        \item Tracelessness
    \end{enumerate}
    and whose Lie bracket $[\,\cdot\,,\,\cdot\,]$ is given by the matrix commutator.
\end{definition}  

\begin{lemma}
\label{lemma:Tao}
    If $G$ is a compact connected Lie group with Lie algebra $\mathfrak{g}$, then the exponential map $\exp\colon\mathfrak{g}\to G$ is surjective.
\end{lemma}
\begin{proof}
  See corollary 11.10 in \cite{Hall2015} for details.
\end{proof}

The above lemma guarantees that every group element of $\SU(n)$ can be expressed as the exponential of some element of its Lie algebra $\su(n)$. It is not self-evident that subspaces of $\SU(n)$ inherit this property. In this section, we provide the mathematical formalism to identify invariant subspaces of $\SU(n)$ that inherit properties of $\SU(n)$ such as connectedness.\\

Throughout the paper we will deal with a special case of the special unitary group, namely $\SU(2^N)$ for $1\leq N\in \mathbb{N}$. The dimension of the corresponding real Lie algebra $\mathfrak{su}(2^N)$ is $4^N-1$. This means that if we want to construct a basis of $\mathfrak{su}(2^N)$, we have to find a total of $4^N-1$ matrices in $\text{M}_{2^N}(\mathbb{C})$ that are $\mathbb{R}-$linearly independent, traceless and skew-Hermitian, as follows from the definition \ref{def:su(N)}. To do this, it will be useful to use the Pauli matrices and the notion of the tensor product of vector spaces, both of which we define in the following:
    
    \begin{definition}
        We define the set of \textbf{extended Pauli matrices} $\{\sigma_i\}_{i=0}^3$ as:

        \begin{equation}
        \sigma_{0}=\left(\begin{array}{ll}1 & 0 \\ 0 & 1\end{array}\right), \quad \sigma_{1}=\left(\begin{array}{ll}0 & 1 \\ 1 & 0\end{array}\right), \sigma_{2}=\left(\begin{array}{cc}0 & -i \\ i & 0\end{array}\right), \quad \sigma_{3}=\left(\begin{array}{cc}1 & 0 \\ 0 & -1\end{array}\right),
        \end{equation}

        \noindent and we will refer to them simply as \textbf{Pauli matrices}.
    \end{definition}

\begin{definition}
    The tensor product of Pauli matrices is a Pauli string,
    \begin{equation}
        \bigotimes_{i=1}^n \sigma^{(i)} := \sigma^{(1)} \otimes \sigma^{(2)} \otimes \cdots \otimes \sigma^{(n)},
    \end{equation}
    \noindent where each $\sigma^{(i)}$ with $i\in\{1,\ldots, n\}$ in the string can be one of $\{\sigma_0, \sigma_1, \sigma_2, \sigma_3\}$.
\end{definition}

These Pauli strings provide a basis of $\mathfrak{su}(2^N)$, upon the choice of lexicographic ordered tensor product. Further details can be found in appendix \ref{sec:lexicographic-order}.

\begin{remark}
    The Pauli strings form a real basis of $\mathfrak{su}(2^N)$ by setting the tensor product $\otimes$ to be the Kronecker product:
    
    \begin{equation}
        \su(2^N) = \operatorname{span}_{\mathbb{R}} \left\{ i\bigotimes_{i=1}^n \sigma^{(i)}\right\} \backslash \{i\sigma_0^{\otimes N}\}.
    \end{equation}

    \noindent where $\sigma_0^{\otimes N}$ is the tensor product of $\sigma_0$ with itself a total of $N$ times.
\end{remark}

\subsection{Restricted subgroups and their corresponding algebras}\label{ssec:restricted-subgroups}

    The unitary group $\SU(2^N)$ has been studied extensively in the context of quantum computing since the action of quantum gates acting on $N$ qubits can be represented as $\SU(2^N)$ matrices acting on $2^N$ dimensional complex unit-vectors, see \cite{nielsen_quantum_2010} for details. However, as discussed in the previous section, one needs a total of $4^N-1$ matrices to form a basis for its Lie algebra (which suffice to fully describe the group), which makes this group computationally expensive to work with or simulate. A common technique in physics to reduce the number of degrees of freedom of a problem is to exploit intrinsic symmetries of the system. In this section we will study subsets of $\SU(2^N)$ that are invariant under the action of certain set of matrices $\mathcal{M}\subseteq \operatorname{M}_{2^N}(\mathbb{C})$. Before arriving at that point, we need to define the subset of $\SU(2^N)$ that is invariant under a set of matrices $\mathcal{M}\subseteq \operatorname{M}_{2^N}(\mathbb{C})$, denoted by $\mathcal{M}\SU(2^N)$:
    
    \begin{equation}
    \label{eq:definvariantSU}
        \mathcal{M}\SU(2^N):= \left\{ A \in \SU(2^N): A \circ B=B \circ A, \forall B \in \mathcal{M} \right\}.
    \end{equation}
    
    The operation denoted by $\circ$ is the usual matrix product, which we will denote in the sequel by juxtaposition as $AB$. Now, we will show that $\mathcal{M}\SU(2^N)$ is a Lie subgroup of $\SU(2^N)$ using Cartan's closed subgroup theorem. 
    
    \begin{theorem}[Closed subgroup theorem]
        Let $G$ be a Lie group and $H \subseteq G$ be an algebraic subgroup that is also a closed subset of $G$. Then $H$ is an embedded Lie subgroup.
    \end{theorem}
    \begin{proof}
        See Theorem 20.12 in \cite{Lee2012} for details.
    \end{proof}
    
    \begin{lemma}
    \label{lemma:subgroup}
        $\mathcal{M}\SU(2^N)$ is an algebraic subgroup of $\SU(2^N)$ .
    \end{lemma}
    \begin{proof}
        Since $\mathcal{M}\SU(2^N)$ is a non-empty subset of $\SU(2^N)$ (note that the identity matrix commutes with any element), it only remains to show that it is closed under the group operation, which in this case is matrix multiplication. Let $A,B \in \mathcal{M}\SU(2^N)$ and $C \in \mathcal{M}$, then:
            \begin{equation}
                C  \left(A   B \right) C^{-1}= C  A  \left( C C^{-1} \right)  B C^{-1}
            \end{equation} 
            \begin{equation}
                = \left( C A  C^{-1}\right)  \left( C B  C^{-1}\right)= A B,
            \end{equation}
        and 
        \begin{equation}
            C A^{-1}C^{-1}= C A^{-1} C^{-1} \left( A A^{-1}\right)= C A^{-1} C^{-1} \left(C A C^{-1}\right) A^{-1}= A^{-1}.
        \end{equation}
        Thus, $\mathcal{M}\SU(2^N)$ is closed under products and inverses and is an algebraic subgroup of $\SU(2^N)$.
    \end{proof}
    
    \begin{lemma}
        \label{lemma:topoclosed}
         $\mathcal{M}\SU(2^N)$ is a topologically closed subset of $\SU(2^N)$. 
    \end{lemma}
    \begin{proof}
        Consider the usual determinant map 
         \begin{equation}
         	\det\colon \operatorname{GL}(2^N,\mathbb{C})\to \mathbb{C},
        \end{equation}
         and the two maps $f_1$ and $f_U$ defined as
         \begin{align}
         f_1\colon\operatorname{GL}(2^N,\mathbb{C})&\to\operatorname{GL}(2^N,\mathbb{C})\nonumber\\
        A &\to AA^\dagger,\\
         f_C \colon \operatorname{GL}(2^N,\mathbb{C})&\to\operatorname{GL}(2^N,\mathbb{C})\nonumber\\
        A &\to CA - AC,
         \end{align}
        for every  $C \in \mathcal{M}$. Then, we may view $\mathcal{M}\SU(2^N)$ as
            
            \begin{equation}
                 \mathcal{M}\SU(2^N) = \bigcap_{C \in \mathcal{M}}f_U^{-1}(\{\textbf{0}\}) \;\bigcap\; f_1^{-1}(\{\mathbbm{1}\}) \;\bigcap\; \text{det}^{-1}(\{1\}),
            \end{equation}
        
        \noindent and since the pre-image under a continuous map of a closed set is closed and the intersection of (arbitrarily many) closed sets is closed, the claim follows.
    \end{proof}

    \begin{theorem}
        $\mathcal{M}\SU(2^N)$ is a compact Lie subgroup of $\SU(2^N)$, i.e. an embedded compact sub-manifold.
    \end{theorem}
    \begin{proof}
        From lemmas \ref{lemma:subgroup} and \ref{lemma:topoclosed}, it follows that $\mathcal{M}\SU(2^N)$ is a closed subgroup of $\SU(2^N)$ and thus, by the closed subgroup theorem, $\mathcal{M}\SU(2^N)$ is a Lie subgroup. Moreover, since $\mathcal{M}\SU(2^N)$ is a closed subgroup of the compact Lie group $SU(2^N)$, it is compact.
    \end{proof}

    \begin{remark}
        One may take $\mathcal{M} = \SU(2^N)$, so that $\mathcal{M}\SU(2^N) = Z(\SU(2^N))$ is the center of $\SU(2^N)$, which is a finite group and thus a compact Lie group (it is a zero dimensional manifold). We point out that it is nevertheless not connected.
    \end{remark}

    Being a Lie subgroup is a strong condition since it requires the subspace to be an embedded sub-manifold with a differentiable structure. In particular, this guarantees the existence of a tangent space at every point on the manifold. The tangent space of a Lie group at the identity corresponds to its Lie algebra. We are now interested in the Lie algebra $\mathfrak{M}\su(2^N)$corresponding to the Lie group $\mathcal{M}\SU(2^N)$. Let us consider a parameterized curve $A(\cdot)\colon [0,1]\to \mathcal{M}\SU(2^{N})$ through the identity element  such that $A(0)= \mathbbm{1}$. Linearising at $t=0$ yields:

    \begin{equation}
        A(t)= A(0)+ t \left. \frac{dA(t)}{dt} \right|_{t=0}+ \mathcal{O}(t^2),
    \end{equation}
    
    where by definition, $\left. \frac{dA(t)}{dt} \right|_{t=0}\in \mathfrak{M} \su(2^N)$. Imposing the $ \mathcal{M}$-invariance condition on both sides of the above equation leads to the condition that the elements of the Lie algebra $\mathfrak{M}\su(2^N)$ must also verify:
    
    \begin{align}
    \label{eq:definvariantalgebra}
        \mathfrak{M}\su(2^N) &= \left\{ a \in \su(2^N):  Ba = aB, \quad \forall B \in  \mathcal{M} \right\} = \nonumber\\ &= \left\{ a \in \su(2^N):  B aB^{-1} = a, \quad \forall B \in  \mathcal{M} \right\}:= \mathcal{M}\su(2^N)
    \end{align}
    
    We have chosen the same notation as for the Lie group $\mathcal{M}\SU(2^N)$, since the definition is completely analogous. It is straightforward to verify that $\mathcal{M}\su(2^N)$ is a Lie subalgebra of $\su(2^N)$, i.e. closed under the Lie bracket.

    \begin{proposition}
        The Lie Algebra $\mathcal{M}\mathfrak{su}(2^{N})$ as defined above is a Lie subalgebra of $\su (2^N)$, i.e. it is algebraically closed under the Lie bracket of $\mathfrak{su}(2^{N})$.
    \end{proposition}
    \begin{proof}
        % Dados A, B in Gsu, y dada alpha cualquiera en S_N, quiero ver
        % que rho(alpha, [A,B]) in Gsu
        %
        %
        Consider two elements $a, b \in \mathcal{M}\mathfrak{su}(2^N) $ and fix any $C \in \mathcal{M}$:
        \begin{align}
            C [a,b]C^{-1} &= CabC^{-1}-CbaC^{-1}\\
            &= Ca \left( C^{-1}C\right)bC^{-1}-Cb \left( C^{-1}C\right)aC^{-1}\\
             &= \left( C a C^{-1}\right) \left( C b C^{-1}\right)-\left( C b C^{-1}\right) \left( C a C^{-1}\right)\\
             &= ab-ba=[a,b],
        \end{align}
    since this holds for any $C \in \mathcal{M}$, $\mathcal{M}\mathfrak{su}(2^N)$ is a Lie subalgebra of $\su (2^N)$.
    \end{proof}
    
    It is however not guaranteed that any element of $\mathcal{M}\SU(2^N)$ can be written as the exponential of an element in the $\mathcal{M}\su(2^N)$ Lie algebra. This is because if $\mathcal{M}\SU(2^N)$ is not connected, the Lie algebra $\mathcal{M}\su(2^N)$ can only generate the connected component of $\mathcal{M}\SU(2^N)$ containing the identity. Below, we provide a theorem indicating for which sets of matrices $\mathcal{M}$, $\mathcal{M}\SU(2^N)$ is connected and thus, the exponential map from the Lie subalgebra $\mathcal{M}\su(2^n)$ to the Lie subgroup $\mathcal{M}\SU(2^N)$ is surjective. 
    
\begin{theorem}
\label{thm:connectivitySU}
Given a subset $\mathcal{M}\subseteq M_{2^N}(\mathbb{C})$, the $\mathcal{M}$-invariant Lie subgroup $\mathcal{M}\SU(2^N)$ defined in equation \eqref{eq:definvariantSU} is connected if:
    \begin{enumerate}
        \item $\mathcal{M}$ is a (not necessarily finite \emph{\cite{Conrad2014-tc}}) set of commuting and diagonalizable matrices.
        \item $\mathcal{M}$ is a set of matrices of the representation of a permutation group acting on $N$ qubits.
    \end{enumerate}
\end{theorem}
\begin{proof}
    The details of this proof can be found in appendix \ref{sec:proofconnected}.
\end{proof}

\begin{theorem}
\label{thm:expissur}
 For the subsets $\mathcal{M}$ listed in theorem \ref{thm:connectivitySU}, $\mathcal{M}\su(2^N)$, as defined in equation \eqref{eq:definvariantalgebra}, is the Lie algebra of the Lie subgroup $\mathcal{M} \SU(2^N)$, and the exponential map $\exp:  \mathcal{M} \su(2^N) \rightarrow \mathcal{M} \SU(2^N)  $ is surjective.
\end{theorem}
\begin{proof}
    This follows from lemma \ref{lemma:Tao} since $\mathcal{M} \SU(2^N)$ is compact and by theorem \ref{thm:connectivitySU} also connected.
\end{proof}

In particular, the above theorem ensures that for the $\mathcal{M}$ discussed above, any element $C \in \mathcal{M} \SU(2^N)$ can be written as the exponential of an element in $\mathcal{M}\su(2^N)$. This is often overlooked in the existing literature, and many authors fail to consider the connectedness property to ensure the surjectivity of the exponential map. This property becomes relevant when studying applications of $\mathcal{M}$-invariant Lie groups, such as in quantum neural networks \cite{nguyen_theory_2024} and quantum operator controllability \cite{quantumcontrol}, where connectedness is necessary to ensure that any group element can be attained as an exponential of the Lie algebra generators. The aim of this section was to give a rigorous and complete account of the mathematical results which ensure that for some $\mathcal{M}$, the $\mathcal{M}$-invariant Lie groups are completely generated by $\mathcal{M}$-invariant Lie subalgebras. We will now discuss the special cases when $\mathcal{M}$ is a permutation group or a single matrix.

\subsection{Permutation invariant Lie subgroups and subalgebras}

\label{sec:permutationinvariance}

    In this section we will explore the case in which $\mathcal{M}$ is a set of matrices that arise from the representation of a permutation group, which we shall denote by $G$, that acts on the Lie algebra $\su(2^N)$. The resulting group will be called \textbf{permutation invariant subgroup} of $\SU(2^N)$ and denoted $G\SU(2^N)$.
 
    First, let us define the representation of the permutation group on the Lie algebra $\su(2^N)$ using the Pauli basis defined in section \ref{sec:propertiesSU}:
 
    \begin{definition}
    \label{def:rhoalpha}
         Let $G$ be a permutation group on at most $N$ symbols. Its action on $\mathfrak{su}(2^N)$ is defined to be the map
    \begin{align}
        \rho\colon G \times \mathfrak{su} \left(2^{N}\right) &\to \mathfrak{su} \left(2^{N}\right) \\
        (\alpha, a) &\mapsto \rho(\alpha, a),
    \end{align} 
    
        \noindent where $\rho(\alpha,\cdot)$ is linear in the second argument for any $\alpha\in G$ and $\rho(\alpha, a)$ is defined on each basis element as:

    \begin{equation}
        \rho\left(\alpha, \sigma_{\mu_{1}} \otimes \cdots \otimes \sigma_{\mu_{N}}\right):= \sigma_{\mu_{\alpha(1)}} \otimes \cdots \otimes \sigma_{\mu_{\alpha(N)}}.
    \end{equation}
    
    \noindent where each $\mu_i\in\{0,1,2,3\}$ for $i\in\{1,2,\ldots, N\}$.
    
    \end{definition}
    
    \begin{theorem}
    \label{th:Uimplementingsymmetry}
        For each $\alpha \in G$ there exists a unique (permutation) matrix $U_\alpha\in U(2^{N})$\footnote{$U(2^N)$ is the unitary group of $2^N\times 2^N$ matrices.} verifying:

        \begin{equation}
            \label{eq:representationpermutation}
            U_{\alpha}\left(\sigma_{\mu_{1}} \otimes \cdots \otimes \sigma_{\mu_N}\right) U_{\alpha}^{\dagger}=\sigma_{\mu_{\alpha(1)}} \otimes \cdots \otimes \sigma_{\mu_{\alpha(N)}}.
    \end{equation}
    \end{theorem}
    \begin{proof}
        Let $\alpha\in G$ be a permutation. We aim to describe the procedure for implementing the action of $\rho(\alpha, \cdot)$ via a unitary matrix.
         
        Since the matrices in $\su(2)$ are $2\times2$ matrices with complex entries, we may regard them as operators on $\mathbb{C}^2$-vectors, i.e. vectors of two components with complex entries.
        
        Consider the vectors $v^{(i)}:= (\beta^{(i)}_1, \beta^{(i)}_{2})^{T}$ for $i=1,\ldots, N$ and build their tensor product $v:= v^{(1)}\otimes v^{(2)}\otimes \cdots \otimes v^{(N)}$ using the lexicographic convention described in section \ref{sec:propertiesSU}, i.e. the Kronecker product:

        \begin{equation}
            v=\begin{pmatrix}
                \beta^{(1)}_1\beta^{(2)}_1\cdots \beta^{(N-1)}_1\beta^{(N)}_1\\
                \beta^{(1)}_1\beta^{(2)}_1\cdots \beta^{(N-1)}_1\beta^{(N)}_2\\
                \beta^{(1)}_1\beta^{(2)}_1\cdots \beta^{(N-1)}_2\beta^{(N)}_1\\
                \beta^{(1)}_1\beta^{(2)}_1\cdots \beta^{(N-1)}_2\beta^{(N)}_2\\
                \vdots\\
                \beta^{(1)}_1\beta^{(2)}_2\cdots \beta^{(N-1)}_1\beta^{(N)}_1\\
                \beta^{(1)}_1\beta^{(2)}_2\cdots \beta^{(N-1)}_1\beta^{(N)}_2\\
                \beta^{(1)}_1\beta^{(2)}_2\cdots \beta^{(N-1)}_2\beta^{(N)}_1\\
                \beta^{(1)}_1\beta^{(2)}_2\cdots \beta^{(N-1)}_2\beta^{(N)}_2\\
                \vdots\\
                \beta^{(1)}_2\beta^{(2)}_1\cdots \beta^{(N-1)}_1\beta^{(N)}_1\\
                \beta^{(1)}_2\beta^{(2)}_1\cdots \beta^{(N-1)}_1\beta^{(N)}_2\\
                \vdots\\
                \beta^{(1)}_2\beta^{(2)}_2\cdots \beta^{(N-1)}_2\beta^{(N)}_2
            \end{pmatrix}.
        \end{equation}
        
        A generic permutation $\alpha\in G$, will act on the tensor product $v=v^{(1)}\otimes v^{(2)}\otimes \cdots \otimes v^{(N)}$ by means of shifting the order of the vectors:  $v^\alpha:=v^{\alpha(1)}\otimes v^{\alpha(2)}\otimes \cdots \otimes v^{\alpha(N)}$. The action of $\alpha$ yields:

     \begin{equation}
        v^\alpha=\begin{pmatrix}
            \beta^{\alpha(1)}_1\beta^{\alpha(2)}_1\cdots \beta^{\alpha(N-1)}_1\beta^{\alpha(N)}_1\\
            \beta^{\alpha(1)}_1\beta^{\alpha(2)}_1\cdots \beta^{\alpha(N-1)}_1\beta^{\alpha(N)}_2\\
            \beta^{\alpha(1)}_1\beta^{\alpha(2)}_1\cdots \beta^{\alpha(N-1)}_2\beta^{\alpha(N)}_1\\
            \beta^{\alpha(1)}_1\beta^{\alpha(2)}_1\cdots \beta^{\alpha(N-1)}_2\beta^{\alpha(N)}_2\\
            \vdots\\
            \beta^{\alpha(1)}_1\beta^{\alpha(2)}_2\cdots \beta^{\alpha(N-1)}_1\beta^{\alpha(N)}_1\\
            \beta^{\alpha(1)}_1\beta^{\alpha(2)}_2\cdots \beta^{\alpha(N-1)}_1\beta^{\alpha(N)}_2\\
            \beta^{\alpha(1)}_1\beta^{\alpha(2)}_2\cdots \beta^{\alpha(N-1)}_2\beta^{\alpha(N)}_1\\
            \beta^{\alpha(1)}_1\beta^{\alpha(2)}_2\cdots \beta^{\alpha(N-1)}_2\beta^{\alpha(N)}_2\\
            \vdots\\
            \beta^{\alpha(1)}_2\beta^{\alpha(2)}_1\cdots \beta^{\alpha(N-1)}_1\beta^{\alpha(N)}_1\\
            \beta^{\alpha(1)}_2\beta^{\alpha(2)}_1\cdots \beta^{\alpha(N-1)}_1\beta^{\alpha(N)}_2\\
            \vdots\\
            \beta^{\alpha(1)}_2\beta^{\alpha(2)}_2\cdots \beta^{\alpha(N-1)}_2\beta^{\alpha(N)}_2
        \end{pmatrix}.
    \end{equation}

        We now reorder each row of the vector to return it to its original form, i.e. the upper index of $\beta$ increases from $1$ to $N$. This time, however, the sub-indices will not generally be those that result from the initial lexicographic order. Instead, after the application of $\alpha$, each reordered row will either remain as it was before the application of $\alpha$ or correspond to another row of the initial vector $v$, since all possible combinations are listed therein. This allows us to assert that there exists a unique matrix $U_\alpha$ mapping from the original vector to the permuted one. This matrix is unique and unitary because it is a permutation matrix, permuting the rows of the vector $v$:

    \begin{equation}
        v^\alpha = U_{\alpha} v.
    \end{equation}

    If the action on vectors is given by $U_\alpha$, it follows that the action on an operator $A$ is given by $U_\alpha A U_\alpha^{-1}= U_\alpha A U_\alpha^{\dagger}$.
    
    \end{proof}

    \begin{definition}
        Given the representation of the permutation group $G$ in terms of the matrices $\mathcal{G}:=\{U_\alpha\}_{\alpha\in G}$, we can now make use of equation \eqref{eq:definvariantSU} to define the permutation-invariant Lie group $G \SU(2^N)$ as follows:
    
    \begin{equation}
        G\SU(2^N):= \mathcal{G}\SU(2^N)= \left\{ A \in \SU(2^N): U_{\alpha} A U_{\alpha}^\dagger = A, \quad \forall \alpha \in G \right\}\label{eq:pisu-definition}.
    \end{equation}
    \end{definition}

    \begin{definition}
        \noindent The \textbf{permutation invariant subalgebra} $G \mathfrak{su}(2^{N})$ is defined in a similar manner:

    \begin{align}
        G \mathfrak{su}(2^{N})&:= \mathcal{G}\su(2^N)=\left\{ a \in \su(2^N): U_{\alpha} a U_{\alpha}^\dagger = a, \quad \forall \alpha \in G \right\}=\\
        &=\operatorname{Fix}_{G}\left(\mathfrak{su}\left(2^{N}\right)\right):=\left\{a \in \mathfrak{su} \left(2^{N}\right) : \rho(\alpha, a)=a \quad \forall \alpha \in G\right\}.
    \end{align}
    \end{definition}

    %{\color{red}
    %\begin{remark}
        %Since the matrices $U_\alpha$ are permutation matrices with $a_{11}=1$, it follows from theorem \ref{thm:expissur} that the exponential map from $G\su(2^N)$ to $G\SU (2^N)$ is surjective.
    %\end{remark}
    %}

    \begin{example}
         \textbf{Recovering the SWAP Matrix. } We show how the lexicographic procedure leads to the well-known $\operatorname{SWAP}$ matrix. To this end, set $N=2$. Let $v^{(1)}:=(\beta^{(1)}_1, \beta^{(1)}_2)^{T}$ and $v^{(2)}:=(\beta^{(2)}_1, \beta^{(2)}_2)^{T}$ such that:

        \begin{equation}
            v^{(1)}\otimes v^{(2)} = \begin{pmatrix}
                \beta^{(1)}_1\beta^{(2)}_1\\
                \beta^{(1)}_1\beta^{(2)}_2\\
                \beta^{(1)}_2\beta^{(2)}_1\\
                \beta^{(1)}_2\beta^{(2)}_2\\
            \end{pmatrix}.
        \end{equation}

        Consider $G$ to be the symmetric group $S_2:=\{\operatorname{id}, (12)\}$ and let $\alpha = (12)$ act on $v^{(1)}\otimes v^{(2)}$ to get:

        \begin{equation}
           v^{\alpha(1)}\otimes v^{\alpha(2)} = \begin{pmatrix}
                \beta^{(2)}_1\beta^{(1)}_1\\
                \beta^{(2)}_1\beta^{(1)}_2\\
                \beta^{(2)}_2\beta^{(1)}_1\\
                \beta^{(2)}_2\beta^{(1)}_2\\
            \end{pmatrix}.
        \end{equation}

        \noindent Reordering as explained above one obtains
    
        \begin{equation}
            v^{\alpha(1)}\otimes v^{\alpha(2)} = \begin{pmatrix}
                \beta^{(1)}_1\beta^{(2)}_1\\
                \beta^{(1)}_2\beta^{(2)}_1\\
                \beta^{(1)}_1\beta^{(2)}_2\\
                \beta^{(1)}_2\beta^{(2)}_2\\
            \end{pmatrix}.
        \end{equation}

        \noindent The matrix $U_{\rho(\alpha,\cdot)}$ implementing the permutation of the vectors is given by:
    
        \begin{equation}
             v^{\alpha(1)}\otimes v^{\alpha(2)} = \begin{pmatrix}
                \beta^{(1)}_1\beta^{(2)}_1\\
                \beta^{(1)}_2\beta^{(2)}_1\\
                \beta^{(1)}_1\beta^{(2)}_2\\
                \beta^{(1)}_2\beta^{(2)}_2\\
            \end{pmatrix}= \underbrace{\begin{pmatrix}
                1 & 0 & 0 & 0\\
                0 & 0 & 1 & 0\\
                0 & 1 & 0 & 0\\
                0 & 0 & 0 & 1
            \end{pmatrix}}_{:= U_{\rho(\alpha,\cdot)}} \begin{pmatrix}
                \beta^{(1)}_1\beta^{(2)}_1\\
                \beta^{(1)}_1\beta^{(2)}_2\\
                \beta^{(1)}_2\beta^{(2)}_1\\
                \beta^{(1)}_2\beta^{(2)}_2\\
            \end{pmatrix}.
        \end{equation}
        
        From the equation above we can see that the matrix $U_{\rho(\alpha,\cdot)}$ implementing the permutation $\alpha=(12)$ corresponds to the matrix representation of the well-known $\operatorname{SWAP}$ matrix acting on $2$ qubits. Thus, the permutation-invariant subgroup $S_2\SU(2^2)$ corresponds to all $\SU(4)$ operators which commute with the $\operatorname{SWAP}$ gate.
    \end{example}

\subsection{Energy invariance}
\label{sec:EnergyInvariance}
    %{\color{red} Beyond the construction of discrete symmetries via the map from the discrete symmetries to $\operatorname{SWAP}$ actions on the Lie group, it is also possible to define invariant groups with respect to other actions} {\color{blue} What does this mean? what is a SWAP action?}
    In the previous section, we discussed the subgroups of $\SU(2^N)$ which remain invariant under the action of a permutation group $G$. We will now consider invariance with respect to a fixed set of matrices. An example is the energy-conserving symmetry, defined as mapping states to other states with the same Hamming distance, e.g. see  3.5 in \cite{codingtheory}. The operators invariant under this symmetry keep the number of $0$s and $1$s in the state the same, thus preserving the energy.
    
    We develop a way to assign \emph{energy} $1$ to the qubit state $|1\rangle$ and \emph{energy} $0$ to the qubit state $|0\rangle$. In order to achieve this we consider the states $|1\rangle$ and $|0\rangle$ to be the $-1$ and $1$ eigenstates of $\sigma_3$, namely:

    \begin{equation}
        \sigma_3|0\rangle = |0\rangle, \quad  \sigma_3|1\rangle = -|1\rangle.
    \end{equation}

    \noindent To achieve the desired property we can define the \emph{Hamiltonian}

    \begin{equation}
        H := \frac{1}{2}(\sigma_0-\sigma_3) = \begin{pmatrix} 0 & 0\\ 0 & 1 \end{pmatrix}.
    \end{equation}

    \noindent because $H|0\rangle = 0$ and $H|1\rangle = |1\rangle$, i.e. $|0\rangle$ and $|1\rangle$ are eigenstates of $H$ with eigenvalues $0$ and $1$ respectively. For a general state $|\psi\rangle = \alpha|0\rangle + \beta |1\rangle$, its energy is computed as:

    \begin{equation}
        \langle \psi | H |\psi\rangle.
    \end{equation}
    
    In the general framework, when dealing with a system of $N$ qubits, we will consider states of the form $|\lambda_1\rangle\otimes|\lambda_2\rangle\otimes\cdots \otimes |\lambda_N\rangle$ with $\lambda_i\in\{0,1\}$ for $i\in\{1,\ldots, N\}$ to be eigenstates of a general Hamiltonian $H^{(N)}$ whose energy is the sum of the energy of each qubit, namely $\sum_{i=1}^{N}\lambda_i$. This leads us to define

    \begin{equation}
    \label{eq:NqubitHamiltonian}
        H^{(N)}:= H\otimes \underbrace{\sigma_{0}\otimes\cdots\otimes \sigma_{0}}_{N-1\text{ times}} + \sigma_{0}\otimes H\otimes \underbrace{\sigma_{0}\otimes\cdots\otimes \sigma_{0}}_{N-2\text{ times}}+\ldots +\underbrace{\sigma_{0}\otimes\cdots\otimes \sigma_{0}}_{N-1\text{ times}}\otimes H.
    \end{equation}

    \noindent Just as for the case of one qubit, for a general state $|\Psi\rangle$ in the $N-$qubit tensor product space, the energy is computed as:

    \begin{equation}
        \langle\Psi|H^{(N)}|\Psi\rangle.
    \end{equation}

    \begin{definition}
    \label{def:energy_preserving}
        An operator $A$ is \textbf{energy preserving} if for any $N$-qubit state $|\Psi\rangle$,
    \begin{equation}
    \langle\Psi| H^{(N)} |\Psi\rangle = \langle\Psi|A^\dagger H^{(N)} A |\Psi\rangle .
    \end{equation}
    \end{definition}
    
    We are interested in finding the operators $A\in \SU(2^{N})$ that verify the previous definition, which is equivalent to finding the operators $A\in \SU(2^{(N)})$ that commute with $H^{(N)}$, i.e. $AH^{(N)} = H^{N}A$, since the property \ref{def:energy_preserving} must hold for any state $|\Psi\rangle$.
    \begin{definition}
        We denote the set of operators in $\SU(2^N)$ commuting with $H^{(N)}$ as:
        \begin{equation}
            H^{(N)}\SU(2^N) := \mathcal{M}\SU(2^N)=\{A\in \SU(2^{N})\colon AH^{(N)} = H^{(N)}A\},
        \end{equation}
        \noindent where $\mathcal{M}=\{H^{(N)}\}$.
    \end{definition}

    \begin{theorem}
    The set $H^{(N)}\SU(2^N)$ is a connected compact subgroup of $\SU(2^{N})$.
    \end{theorem}
    \begin{proof}
        This result follows from theorem \ref{thm:connectivitySU} since $H^{(N)}$ is diagonalizable (it is already diagonal) and trivially commutes with itself.
    
       % The only non-trivial property is connectedness. To show it, we prove that $A\in H^{(N)}\SU(2^N)$ is connected to the identity. Since $A$ and $H^{(N)}$ commute, they are simultaneously diagonalizable and there exists a basis in which $A$ and $H^{(N)}$ are both diagonal, namely there exists a matrix $P_A$ such that:

       %  \begin{equation}
       %      P_AAP_A^{-1} = \text{diag}(e^{i\theta_1}, \ldots, e^{i\theta_{2^{N}}}) \quad \text{and} \quad P_AH^{(N)}P_A^{-1} = \text{diag}(\ldots):=D_H
       %  \end{equation}
        
       %  \noindent We define for $t\in [0,1]$ the matrix

       %  \begin{equation}
       %      D_A(t) := \text{diag}(e^{it\theta_1}, \ldots, e^{it\theta_{2^{N}}})
       %  \end{equation}

       %  \noindent which verifies $D_A(0)=\mathbbm{1}$ and $D_A(1)=D_A$. Let now $A(t) := P_A^{-1}D_A(t)P_A$. Note that $A(0)=\mathbbm{1}$ and $A(1)= A$. Moreover, for any $t\in (0,1)$, since diagonal matrices commute,

       %  \begin{align}
       %      A(t)H^{(N)} &= P_A^{-1}D_A(t)P_AH^{(N)}= P_A^{-1}D_A(t)P_AP_A^{-1}D_HP_A= P_A^{-1}D_A(t)D_HP_A\\
       %      &= P_A^{-1}D_H D_A(t)P_A = P_A^{-1}D_H P_AP_A^{-1} D_A(t)P_A = H^{(N)}A(t)
       %  \end{align}
        
       %  \noindent i.e. $A(t)\in \SU(2^{N})$ for any $t\in [0,1]$ and it is a path connecting $A$ to $\mathbbm{1}$, so we are done.
    \end{proof}

\section{Computing the dimension of the symmetry restricted Lie subalgebras}\label{sec:scaling}

    In the previous section, we introduced some $\mathcal{M}$-invariant subgroups of $\SU(2^N)$ and their corresponding Lie algebras. In this section, we present an explicit method for computing the dimension of permutation-invariant subalgebras for different permutation groups $G$. In many applications, it is crucial to comprehend the scaling of the underlying Lie algebras. For instance, if the invariant subalgebra corresponds to the generators of an invariant quantum neural network (QNN), i.e. its dynamical Lie algebra, the scaling of the dimension of the algebra can yield information about the trainability of such QNN \cite{ragone_unified_2024}. The methods presented in this section provide a guide to computing the dimension of permutation-invariant tensor products and can be generalized to, for instance, permutation-invariant space constructed via $N$-tensor products.

\subsection{Permutation invariant subalgebras}\label{ssec:scaling-permutation}

    We  calculate the dimension of the $G$-invariant Lie subalgebras $G\su(2^N)$ for the following permutation groups $G$:
    
    \begin{itemize}
        \item $S_N$, the group of permutations on $N$ elements, also called symmetric group.
        \item $A_N$, the alternating group on $N$ elements, which is the subgroup of $S_N$ consisting of even permutations.
        \item $D_N$, the dihedral group on $N$ vertices, or the symmetry group of an $N-$gon.
        \item $C_N$, the cyclic group on $N$ elements.
        \item $E_N$, the trivial group on $N$ elements.
    \end{itemize}

    \noindent Let $G\in \{S_N, A_N, D_N, C_N, E_N\}$ and let $a\in \mathfrak{su}(2^{N})$ be a generic element in the Lie algebra. We can write it in terms of the basis elements as

    \begin{equation}
        a=\sum_{i_1,\ldots, i_N=0}^{3}a^{i_1\ldots i_N}\sigma_{i_1}\otimes \cdots \otimes \sigma_{i_N}, 
    \end{equation}

    \noindent where $a^{i_1\ldots i_N}\in \mathbb{R}$ for all $i_1,\ldots, i_N\in\{0,1,2,3\}$, $a^{0\ldots 0} = 0$ and any element $\alpha \in G$ acts on $a$ via $\rho(\alpha, \cdot)$ as defined in \ref{def:rhoalpha}:

    \begin{equation}
        \rho(\alpha,a) = \sum_{i_1,\ldots, i_N=0}^{3} a^{i_1\ldots i_N} \rho(\alpha,\sigma_{i_1}\otimes \cdots\otimes \sigma_{i_N}) = \sum_{i_1,\ldots, i_N=0}^{3} a^{i_1\ldots i_N}\sigma_{i_{\alpha(1)}}\otimes \cdots \otimes \sigma_{i_{\alpha(N)}}.
    \end{equation}

    \noindent Requiring the element $a$, which is nothing but a tensor, to be invariant under the action of $G$ is equivalent to requiring:
    
    \begin{align}
        \sum_{i_1,\ldots, i_N=0}^{3}  a^{i_1\ldots i_N}\sigma_{i_1}\otimes \cdots \otimes \sigma_{i_N} &\stackrel{!}{=} \sum_{i_1,\ldots, i_N=0}^{3}  a^{i_1\ldots i_N}\sigma_{i_{\alpha(1)}}\otimes \cdots \otimes \sigma_{i_{\alpha(N)}}  \\
        &=\sum_{i_1,\ldots, i_N=0}^{3}  a^{i_{\alpha^{-1}(1)}\ldots i_{\alpha^{-1}(N)}}\sigma_{i_1}\otimes \cdots \otimes \sigma_{i_N} \quad \forall \alpha\in G.
    \end{align}
    
    \noindent Since the Lie algebra is a vector space, if two vectors are equal, their coordinates must be equal, and thus the components of $a$ have to be invariant, that is,
    
    \begin{equation}
    \label{eq:definition_of_invariance}
        a^{i_1\ldots i_N} = a^{i_{\alpha^{-1}(1)} \ldots i_{\alpha^{-1}(N)}}  \quad \forall \alpha\in G  \quad \quad \Longleftrightarrow \quad \quad a^{i_1\ldots i_N} = a^{i_{\alpha(1)} \ldots i_{\alpha(N)}}  \quad \forall \alpha\in G .
    \end{equation}

    \noindent 
      The dimension of the subspace invariant under G can be computed by counting the number of different reorderings of the indices of the tensor component $a^{i_1\ldots i_N}$ that satisfy the property defined in equation \eqref{eq:definition_of_invariance}. The solution to this problem is inspired by the language of necklaces and bracelets from combinatorics, for which we provide some relevant definitions below. We refer the reader to section $1$ in \cite{Necklaces} for more details.

    \begin{definition}
        Let $k,N\in \mathbb{N}$. A $k-$\textbf{alphabet} is a collection of $k$ letters and an $N-$\textbf{string} over a $k-$alphabet is a sequence of $N$ letters taken from said alphabet.
    \end{definition}

    \begin{definition}
        A \textbf{$k-$ary necklace of length $N$} is an equivalence class under rotations of $N-$character strings over an alphabet of length $k$. If the equivalence classes are taken under both rotations and reflections, then we call it a \textbf{$k-$ary bracelet of length $N$}.
    \end{definition}

    \noindent Note that a necklace is invariant under the cyclic group $C_N$, whereas a bracelet is invariant under the dihedral group $D_N$. Thus, we define in a more general context a $G-$necklace:

    \begin{definition}
        Let $k,N\in \mathbb{N}$ and let $G$ be a permutation group on $N$ elements. A \textbf{$G-$necklace over a $k-$alphabet} is an equivalence class under the action of $G$ on an $N-$string over a $k-$alphabet.
    \end{definition}

    \noindent Note that the task of counting the number of $G-$necklaces is the task of computing the number of equivalence classes under the action of $G$.

    We can adapt this language to our current problem by noticing that, for instance, computing the number of tensors $a\in \mathfrak{su}(2^N)$ invariant under the cyclic group $C_N$ is equivalent to counting the number of different $C_N-$necklaces that can be formed with an alphabet of $4$ letters. The correspondence in general is as follows:

    \begin{itemize}
        \item The possible $k$ letters are the $4$ indices of the Pauli matrices: $\{0,1,2,3\}$.
        \item The length of the $N-$string is the number of indices $i_1,\ldots, i_N$, that is, $N$.
    \end{itemize}

    It is known from the theory of combinatorics that counting the number of conjugacy classes is achieved with the help of the cycle index polynomial, in which the unknowns are replaced by the number of available letters on the alphabet \cite{Necklaces}. In our case though, since the element $\sigma_0\otimes \cdots \otimes \sigma_0$ does not belong to the Lie algebra $\su(2^N)$, we need to subtract one to the previous calculation.

    \begin{example}
        We compute the dimension of the Lie subalgebra $C_4\mathfrak{su}(2^{4})$. We know from the theory we have just developed, that we need to use the cycle index polynomial, see\cite{Polya} for details:

        \begin{equation}
            Z[C_4] (a_1,a_2,a_4) := \dfrac{1}{4}(a_1^{4} + a_2^{2} + 2a_{4}),
        \end{equation}

        \noindent in which we replace the variables $a_1,a_2, a_4$ by the number of letters in our alphabet, i.e.  $a_1=a_2=a_4=4$ and evaluate it, yielding:

        \begin{equation}
            Z[C_4](4,4,4) = \dfrac{1}{4}(4^{4} + 4^{2} + 2\cdot 4) = 70.
        \end{equation}

        \noindent Lastly, we subtract $1$ as explained above to get

        \begin{equation}
            \dim_{\mathbb{R}}\{C_4\mathfrak{su}(2^{4})\} = 69.
        \end{equation}
    \end{example}

    \begin{example}
        Consider the Lie subalgebra $E_N\su(2^N)$. Its dimension is computed as before, but this time using $Z[E_N](a_1)=a_1^N$, where $a_1=4$. Subtracting one yields:

        \begin{equation}
            \dim_\mathbb{R}\{E_N\su(2^N)\} = 4^N-1,
        \end{equation}
    
        \noindent which agrees, as expected, with the dimension of $\su(2^N)$, since $E_N$ is the trivial group.
    \end{example}
    
    \noindent In the following theorem we give the recipe to compute the dimension in the case of the groups discussed above.
    \begin{theorem}
    \label{th:dimesnion_subalgebra}
        Let $G \in\{S_N, A_N, D_N, C_N, E_N\}$. The dimension of the real Lie algebra $G\mathfrak{su}(2^{N})$ is given by:

        \begin{equation}
        \label{eq:dimension_subalgebra}
            \dim_{\mathbb{R}}\{G\mathfrak{su}(2^{N})\} = Z[G](\{4,\ldots, 4\}) - 1,
        \end{equation}
        \noindent where the cycle index polynomials for the group $G$ in the unknowns $\{a_i\}$, which we denote by $Z[G](\{a_i\})$, are:

        \begin{align}
            Z[S_N](\{a_i\}) &= \dfrac{1}{N}\sum_{l=1}^Na_lZ[S_{N-l}] \quad \text{ with } Z[S_0] :=1,\\
            Z[A_N] (\{a_i\}) &= Z[S_N](\{a_i\}) + Z[S_N](\{(-1)^{i-1}a_i\}),\\
            Z[D_N](\{a_i\}) &= \begin{cases}
                \dfrac{1}{2}Z(C_N) + \dfrac{1}{4}(a_1^2a_2^{\frac{N-2}{2}}+ a_2^{\frac{N}{2}}) & \quad \text{for}\  N \text{even},\\
                & \\
                \dfrac{1}{2}Z(C_N) + \dfrac{1}{2}a_1a_2^{\frac{N-1}{2}} & \quad \text{for}\ N \text{ odd},
            \end{cases}\\
            Z[C_N](\{a_i\}) &= \sum_{d|N}\phi(d)a_d^{\frac{N}{d}},\\
            Z[E_{N}](\{a_j\}) &= a_1^N.
        \end{align}
    \end{theorem}

    \begin{remark}
        The procedure we have just described is also useful for computing the dimension of the vector subspace of the tensors with $N$ indices taking values on $\{1,\ldots,k\}$ that are invariant under a permutation group $G$ on $N$ elements. One can use the previous theorem, but instead of substituting $4$ into the cycle index polynomial, one would have to use $k$.
    \end{remark}

    \begin{theorem}
        The real dimension of $S_N\su(2^N)$ is given by:

        \begin{equation}
            \dim_\mathbb{R} \{S_N\su(2^N)\} = \binom{N+3}{N}-1.
        \end{equation}
    \end{theorem}
    \begin{proof}
As described before, we are interested in the number of different $S_N$-necklaces that can be formed with an alphabet of 4 letters. This problem is equivalent to finding the number of ways of writing $N$ as the sum of an ordered sequence of four integers, i.e. $a+b+c+d=N$ where $a,b,c,d \in \mathbb{Z}$, known as the \textit{integer weak composition} problem.  This is obtained by computing the number of ways to group $N$ balls using $3$ bars, see I.3 in \cite{combinatorics} for details, given by
\begin{equation}
\binom{N+3}{N},
\end{equation}
subtracting $1$ as indicated above, we obtain the desired result.
    \end{proof}
    
    We now study the scaling of $G\su(2^N)$ for different number of qubits and different groups. An overview of the growth of different symmetry-restricted subalgebras and their corresponding groups is given in table \ref{tab:scalings} and represented graphically in figure \ref{fig:graphic1}.%The results are included in the following table and represented in the subsequent graph:

%    \begin{table}[hbtp]
%        \centering
%    
%    \begin{tabular}{r | l l}
%        Lie Algebra & Growth of Symmetry & Scaling of $\dim_\mathbb{R} \{G\su(2^N)\}$ \\\hline
%        
%        $\su(2^N)$ & $\mathcal{O}(1)$ & $\mathcal{O}(4^N)$\\
%        $C_N\su(2^N)$ & $\mathcal{O}(N)$ & $\mathcal{O}(\frac{4^N}{N})$\\
%        $A_N\su(2^N)$  & $\mathcal{O}(N!)$ & $\mathcal{O}(N^3)$\\
%        $S_N\su(2^N)$ & $\mathcal{O}(N!)$ & $\mathcal{O}(N^3)$\\\hline
%    \end{tabular}
%
%    \caption{Overview of the growth of the dimensions of the subalgebras for different symmetry restrictions.}
    
%    \label{tab:scalings1}
%\end{table}

    \begin{figure}[hbtp]
        \centering
        \begin{tikzpicture}
\begin{semilogyaxis}[small,
height=7cm, width=12cm,
xlabel = {Number of qubits},
axis x line = bottom,
axis y line = left,
ymajorgrids,
major grid style = {very thin, gray!50},
major tick style = {very thin, gray!50},
%ytick align=outside,
axis line style={gray},
axis line style={gray},
axis line shift=2pt,
xmin = 0.9,
xmax = 14.1,
ylabel = {$\dim_{\mathbb{R}} (G\su(2^N))$},
ymax = 4.5e8,
%legend entries = {permutation symmetry, cyclic symmetry, cyclic symmetry 2, standard ansatz, strongly entangled standard ansatz},
%legend pos=south east,
legend style={nodes={scale=0.8, transform shape}}, 
legend pos=north west]
%
%%% data needs to be in column format
% epch, Sn, Cn, Cn2, entanglement, strong entanglement
% 1, 0.1, 0.2, 0.3, 0.4, 0.5
% 2, 0.1, 0.2, 0.3, 0.4, 0.5
%
% and so on
\addplot [teal] table [x=N, y=C_N,  col sep=semicolon, /pgf/number format/read comma as period]{DatenGruppen.csv};
\addplot [orange] table [x=N, y=D_N, col sep=semicolon, /pgf/number format/read comma as period]{DatenGruppen.csv};
\addplot [olive] table [x=N, y=A_N, col sep=semicolon, /pgf/number format/read comma as period]{DatenGruppen.csv};
\addplot [red, dashed] table [x=N, y=S_N, col sep=semicolon, /pgf/number format/read comma as period]{DatenGruppen.csv};
\addplot [purple] table [x=N, y=Unrestricted, col sep=semicolon, /pgf/number format/read comma as period]{DatenGruppen.csv};

\draw (axis cs: 12.9, 8e7) node[anchor = south east, fill=white, inner sep =0pt] {unrestricted}
	(axis cs:13, 7e6) node[anchor=south] {$C_N$}
	(axis cs:13.3, 4e5) node[circle, fill=white, inner sep=0pt, anchor=south] {$D_N$}
	(axis cs:13.3, 8e2) node[fill=white, inner sep=0pt, anchor=south] {$A_N$}
	(axis cs:13.3, 1e2) node[circle, fill=white, inner sep=0pt, anchor=south] {$S_N$};

%\legend{$C_N$,$D_N$,$A_N$,$S_N$, $E_N$ (unres.)}
\end{semilogyaxis}
%\draw[red] (0,0) rectangle (12, 8);
\end{tikzpicture}
        \caption{Comparison of the scaling of different symmetry-restricted Lie algebras, given by their restriction. The unrestricted line corresponds to the maximum Lie algebra with a scaling of $\mathcal{O}(4^N)$. $S_N$ and $A_N$ overlap, since they have the same scaling structure. The scaling as a function of $N$ is also given in table \ref{tab:scalings}.}
        \label{fig:graphic1}
    \end{figure}
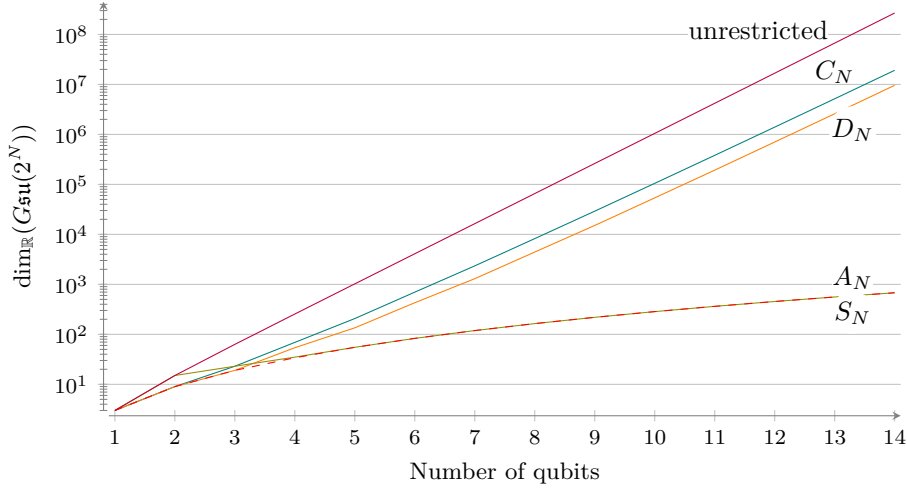

\subsection{Other permutation subgroups}

    We now discuss the dimension of the Lie subalgebra invariant under a permutation group which acts on a subset of $0<N'\leq N$ qubits or in which the permutation can act on all $N$ qubits, but in different ways, depending on the subset of qubits which is considered. Since the former case is just a corollary of the latter, we will discuss this one in detail.
    
    We partition the $N$ qubits into subsets and require each subset be invariant under a particular permutation group. To this end, let $(\lambda_1,\ldots, \lambda_k)$ be a partition of $N$, i.e.

    \begin{equation}
        \lambda_1\geq \lambda_2\geq \ldots, \geq \lambda_k \quad \text{and} \quad \sum_{i}\lambda_i =N.
    \end{equation}

    \noindent Just as before, let $a\in \su(2^N)$, whose coordinates we now write as

    \begin{equation}
        a^{i_1\ldots i_{\lambda_1}i_{\lambda_1+1}\ldots i_{\lambda_1+\lambda_2}i_{\lambda_1+\lambda_2+1}\ldots i_{\lambda_1+\lambda_2+\lambda_3}\ldots i_{\lambda_1+\lambda_2+\ldots+\lambda_{k-1}+1}\ldots i_N}.
    \end{equation}
    
    \noindent and define the following subsets

    \begin{align}
        S_1 &:= \{1,\ldots, \lambda_1\}, \\
        S_2 &:= \{\lambda_1+1,\ldots, \lambda_1+\lambda_2\},\\
        S_3 &:= \{\lambda_1+\lambda_2+1,\ldots, \lambda_1+\lambda_2+\lambda_3\},\\
        &\vdots\\
        S_k &:= \{\lambda_1+\lambda_2+\ldots \lambda_{k-1}+1, \ldots, n\},
    \end{align}

    \noindent and a group $G:=G_1\times G_2\times \ldots \times G_k$, where each $G_i\in\{S_{\lambda_i}, A_{\lambda_i}, D_{\lambda_i}, C_{\lambda_i}, E_{\lambda_i}\}$ acts on the set $S_i$, permuting its $\lambda_i-$elements.

    We now wish to compute the dimension of $G\su(2^N)$, but the situation is not very different from that of the previous section. We can use Theorem \ref{th:dimesnion_subalgebra} to compute the dimension of the space invariant under each $G_i$, having the other indices fixed, so that the product of the dimensions of each $G_i$ invariant subspace minus one (we still have to discard $\sigma_0\otimes \ldots \otimes \sigma_0$) gives the dimension of $G\su(2^N)$. Note, however, that we now may have chains of $\sigma_0\otimes \ldots \sigma_0$ of length $\lambda_i$, as long as we do not have one of length $N$. Thus, in theorem \ref{th:dimesnion_subalgebra} we set $N=\lambda_i$ and do not subtract one on equation \ref{eq:dimension_subalgebra}. We then multiply the dimensions obtained for each $G_i$ and subtract one. We summarize this in the following theorem:

    \begin{theorem}
    \label{th:dimesnion_subalgebra_fancy}
        Let $k\in \mathbb{N}$ and $(\lambda_1,\lambda_2, \ldots, \lambda_k)$ be a partition of $N$ satisfying $\lambda_1\geq\lambda_2\geq\ldots \lambda_k$ and $\sum_{i=1}^{k}\lambda_i=N$. Consider a group $G=G_1\times G_2\times \ldots \times G_k$ where $G_i \in\{S_{\lambda_i}, A_{\lambda_i}, D_{\lambda_i}, C_{\lambda_i}, E_{\lambda_i}\}$. The dimension of $G\mathfrak{su}(2^{N})$ seen as a real Lie algebra is given by \cite{harary1969graph}:

        \begin{equation}
        \label{eq:dimension_subalgebra_fancy}
            \dim_{\mathbb{R}}\{G\mathfrak{su}(2^{N})\} = \prod_{i=1}^{k} Z[G_i](\{4,\ldots, 4\}) - 1 = \prod_{i=1}^{k}\left(\dim_\mathbb{R}\{G_i\su(2^{\lambda_i})\}+1\right) - 1,
        \end{equation}
        \noindent where the cycle index polynomials are given by:

        \begin{align}
            Z[S_{\lambda_i}](\{a_j\}) &= \dfrac{1}{\lambda_i}\sum_{l=1}^{\lambda_i}a_lZ[S_{\lambda_i-l}] \quad Z[S_0] :=1,\\
            Z[A_{\lambda_i}] (\{a_j\}) &= Z[S_{\lambda_i}](\{a_j\}) + Z[S_{\lambda_i}](\{(-1)^{j-1}a_j\}),\\
            \vspace{3mm}
            Z[D_{\lambda_i}](\{a_j\}) &= \begin{cases}
                \dfrac{1}{2}Z(C_{\lambda_i}) + \dfrac{1}{4}(a_1^2a_2^{\frac{\lambda_i-2}{2}}+ a_2^{\frac{\lambda_i}{2}}) & \quad \text{for} \quad \lambda_i \text{ even},\\
                & \\
                \dfrac{1}{2}Z(C_{\lambda_i}) + \dfrac{1}{2}a_1a_2^{\frac{\lambda_i-1}{2}} & \quad \text{for} \quad \lambda_i \text{ odd},
            \end{cases}\\
            Z[C_{\lambda_i}](\{a_j\}) &= \sum_{d|\lambda_i}\phi(d)a_d^{\frac{\lambda_i}{d}},\\
            Z[E_{\lambda_i}](\{a_j\}) &= a_1^{\lambda_i}.
        \end{align}
    \end{theorem}

    \begin{example}
        Take $M\in\mathbb{N}$ and let $N=2M$. Consider $G = S_2\underbrace{\times\ldots \times}_{M \text{ times}}S_2$. Then:

        \begin{equation}
            \dim_\mathbb{R}\left\{G\mathfrak{su}(2^{N}) \right\}= \left[\binom{5}{3}\right]^M -1 = 10^{M}-1.
        \end{equation}
    
        Note that the dimension of this subspace is smaller than the dimension of $\mathfrak{su}(2^N)$, which would be $4^{N}-1=16^{M}-1$, but larger that that of $S_N\mathfrak{su}(2^N)$.
    
        We could have also considered $G = S_M\times S_M$. Then:
    
        \begin{equation}
            \dim_\mathbb{R}\left\{G\mathfrak{su}(2^{N}) \right\} =  \left[\binom{M+3}{3}\right]^2 -1 = \left[\dfrac{(M+3)(M+2)(M+1)}{6}\right]^{2}-1.
        \end{equation}
        This is roughly the order of $\mathcal{O}(M^6)=\mathcal{O}((\frac{1}{2}N)^6)=\mathcal{O}(N^6)$, whereas for $S_N\mathfrak{su}(2^N)$ it would be $\mathcal{O}(N^3)$.
    \end{example} 

\begin{figure}[hbtp]
    \centering
    \pgfdeclareplotmark{octagon*}{
\fill (0:2pt) \foreach \angle in {45, 90, ..., 315} {--(\angle:2pt)};}
\begin{tikzpicture}
\begin{semilogyaxis}[small,
height=8cm, width=12cm,
xlabel = {Number of qubits},
axis x line = bottom,
axis y line = left,
ymajorgrids,
major grid style = {very thin, gray!50},
major tick style = {very thin, gray!50},
%ytick align=outside,
axis line style={gray},
axis line shift=2pt,
xmin = .9,
xmax = 14.1,
ylabel = {$\dim_{\mathbb{R}} (G\su(2^N))$},
ymax = 3.7e8,
%legend entries = {permutation symmetry, cyclic symmetry, cyclic symmetry 2, standard ansatz, strongly entangled standard ansatz},
%legend pos=south east,
legend style={nodes={scale=0.8, transform shape}}, 
legend pos=north west]
%
%%% data needs to be in column format
% epch, Sn, Cn, Cn2, entanglement, strong entanglement
% 1, 0.1, 0.2, 0.3, 0.4, 0.5
% 2, 0.1, 0.2, 0.3, 0.4, 0.5
%
% and so on
\addplot [mark=triangle*, teal, ] table [x=N, y=C3,  col sep=semicolon, /pgf/number format/read comma as period]{DatenGruppenMix.csv};
\addplot [mark=pentagon*, teal] table [x=N, y=C5,  col sep=semicolon, /pgf/number format/read comma as period]{DatenGruppenMix.csv};

\addplot [mark=octagon*, teal, ] table [x=N, y=C8,  col sep=semicolon, /pgf/number format/read comma as period]{DatenGruppenMix.csv};

%%%%%%%%%%

\addplot [mark=triangle*, orange] table [x=N, y=D3,  col sep=semicolon, /pgf/number format/read comma as period]{DatenGruppenMix.csv};
\addplot [mark=pentagon*, orange] table [x=N, y=D5,  col sep=semicolon, /pgf/number format/read comma as period]{DatenGruppenMix.csv};

\addplot [mark=octagon*,orange] table [x=N, y=D8,  col sep=semicolon, /pgf/number format/read comma as period]{DatenGruppenMix.csv};

%%%%%%%%%%

\addplot [mark=triangle*, red] table [x=N, y=S3,  col sep=semicolon, /pgf/number format/read comma as period]{DatenGruppenMix.csv};
\addplot [mark=pentagon*, red] table [x=N, y=S5,  col sep=semicolon, /pgf/number format/read comma as period]{DatenGruppenMix.csv};

\addplot [mark=octagon*, red] table [x=N, y=S8,  col sep=semicolon, /pgf/number format/read comma as period]{DatenGruppenMix.csv};

\addplot [purple] table [x=N, y=Unrestricted, col sep=semicolon, /pgf/number format/read comma as period]{DatenGruppen.csv};
\addplot [red] table [x=N, y=S_N, col sep=semicolon, /pgf/number format/read comma as period]{DatenGruppen.csv};

\legend{$C_{N-3}\times E_3$,
$C_{N-5}\times E_5$,
$C_{N-8}\times E_8$,
$D_{N-3}\times E_3$,
$D_{N-5}\times E_5$,
$D_{N-8}\times E_8$,
$S_{N-3}\times E_3$,
$S_{N-5}\times E_5$,
$S_{N-8}\times E_8$,
$E_N$ (unres.),
$S_N$ }
\end{semilogyaxis}
%\draw[red] (0,0) rectangle (12, 8);
\end{tikzpicture}
    \caption{Comparison of the order of some symmetry restricted Lie algebras. The alternating group has the same behaviour as the permutation group, so it is not included here. For reference, we include the unrestricted ($E_N\su(2^N)$) and the $S_N\su(2^N)$ groups.}
    \label{fig:graphic2}
\end{figure}
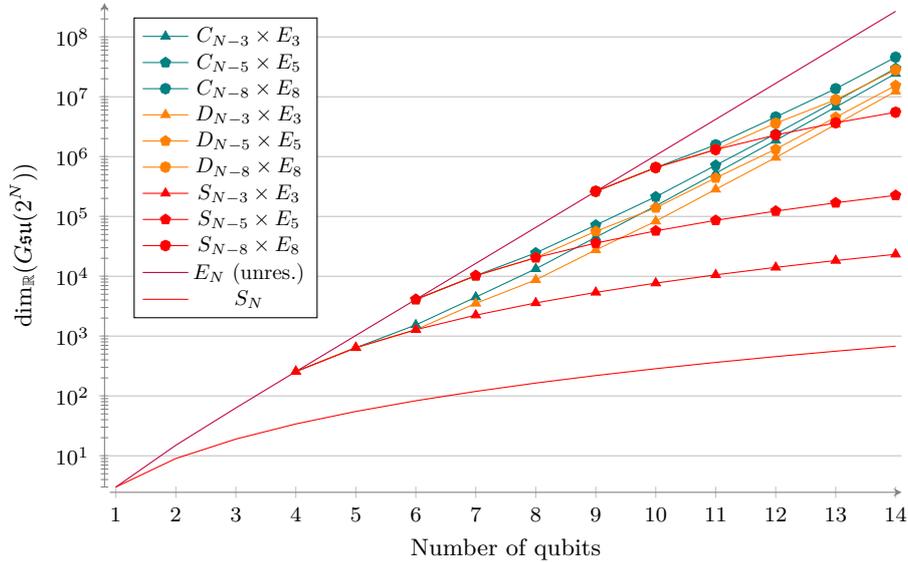

%    \begin{example}
%        We consider the simplest example of the groups that act on $N'=N-\mu$ qubits for $\mu\in\{3,5,8\}$, namely $G_{N-\mu}\times E_{\mu}$        
%    \end{example}

\subsection{Energy-preserving subalgebra}
    In the following, we are interested in the scaling of the dimension with respect to the number of qubits $N$ of the subgroup $H^{(N)}\SU{(2^N)}$, introduced in section \ref{sec:EnergyInvariance}. This can be computed easily if we consider a specific form of the elements $A \in H^{(N)}\SU{(2^N)}$. Choosing the appropriate basis, $H$  can be diagonalized and brought to an ordered form, where
\begin{equation}
H = \displaystyle \bigoplus _{i=1}^kn_iI_i,\,\,\,\,\,\,\,\, n_1 > n_2> \ldots > n_k.
\label{energyordered}
\end{equation} 
Suppose a matrix $A$ commutes with $H$. The matrix $A$ can be written as a block matrix, i.e.
    \begin{equation}
        A=\begin{pmatrix} 
        A_{11} & A_{12} & \cdots & A_{1k} \\
        A_{21} & A_{22} & \cdots & A_{2k} \\
        & & \cdots & \\
        A_{k1} & A_{k2} & \cdots & A_{kk} \\
        \end{pmatrix}.
    \end{equation}
    The size of $A_{jj}$ is equal to the size of $I_{j}$ in equation \eqref{energyordered}. Then
    \begin{equation}
        (HA)_{ij} = n_i A_{ij}, (AH)_{ij} = n_j A_{ij}.
    \end{equation}
    When $i\neq j$, since $HA=AH$ and $n_i\neq n_j$, we get that $(n_i-n_j)A_{ij} = 0\Rightarrow A_{ij}=0$. Therefore, $A$ is a diagonal block matrix. Using this we prove the following theorem regarding the dimension of the energy-preserving subspace of $\SU{(2^N)}$:
\begin{theorem}
The dimension of the energy preserving subspace on $N$ qubits, $H^{(N)}\SU{(2^N)}$ is
\begin{equation}
    \dim_\mathbb{R}\{H^{(N)}\SU{(2^N)}\}=\binom{2N}{N}-1\approx \mathcal{O}\left(\frac{4^N}{\sqrt{N}}\right).
\end{equation}
    \begin{proof}
        In the following proof, we choose to work in the energy-ordered basis where $H^{(N)}$ is given by equation \eqref{energyordered}. We start by noticing that the dimension of the eigenspace $i$ is $\dbinom{N}{i}\times \dbinom{N}{i}$, which corresponds to $i$ qubits being in a state of energy 1 and $N-i$ qubits in energy 0.
        
        Let $A \in H^{(N)}\text{SU}(2^N)$. In the energy-ordered basis, $A$ is block diagonal, of the form $A = \bigoplus_i A_i$.  Moreover, from the unitarity of $A$ we obtain that:
    
    \begin{equation}
        A^\dagger A = AA^\dagger = \bigoplus_i A_iA_i^\dagger=\bigoplus_i\mathbbm{1}_i.
    \end{equation}

    This implies that $A_i^\dagger A_i=A_iA_i^\dagger = \mathbbm{1}$ for $i=1,\ldots, N$, i.e. each block is unitary. The matrix representation of $A$ in the energy-ordered basis is of the form,
    \begin{equation}
        A = \begin{pmatrix}
            A_0 & & & \\
            & A_1 & & \\
            & & \ddots & \\
            & & & e^{-i \phi}A_N
        \end{pmatrix}.
    \end{equation}
    where, $A_i \in \text{U}(\binom{N}{i})$ and $\prod_i \det(A_i) = e^{i \phi}$. The factor of $e^{-i \phi}$ ensures that  $A \in \text{SU}(2^N)$.

    Therefore, 

    \begin{equation}
        \dim_{\mathbb{R}}\{H^{(N)}\SU{(2^N)}\}=\sum_{i=0}^N\dbinom{N}{i}^2-1=\binom{2N}{N}-1,
    \end{equation}
    where 1 is subtracted because $\det(A)=1$. Using Stirling's formula, $N! \sim \sqrt{2 \pi N} (N/e)^N$, we obtain the scaling of the dimension as,
\begin{equation}
    \dim_{\mathbb{R}}\{H^{(N)}\SU{(2^N)}\} \sim \frac{4^N}{\sqrt{N}}.
\end{equation}
    \end{proof}
\end{theorem}

\begin{example}
    \textbf{Energy preserving symmetry on two qubits}. We consider the Hamiltonian for a two-qubit system given by equation \eqref{eq:NqubitHamiltonian}, which in this particular case reduces to:

    \begin{equation}
        H^{(2)} := \begin{pmatrix} 0 & 0\\ 0 & 1 \end{pmatrix} \otimes \begin{pmatrix} 1 & 0\\ 0 & 1 \end{pmatrix} + \begin{pmatrix} 1 & 0\\ 0 & 1 \end{pmatrix} \otimes \begin{pmatrix} 0 & 0\\ 0 & 1 \end{pmatrix}= \begin{pmatrix}
            0 & 0 & 0 & 0\\
            0 & 1 & 0 & 0\\
            0 & 0 & 1 & 0\\
            0 & 0 & 0 & 2
        \end{pmatrix}.
    \end{equation}
    Note that $H^{(2)}$ is already in the energy-ordered form.

    Let us now look at a generic matrix $A\in \operatorname{M}_4(\mathbb{C})$ and impose that it commutes with $H^{(2)}$, namely:

    \begin{equation}
        \begin{pmatrix}
            a_{11} & a_{12} & a_{13} & a_{14}\\
            a_{21} & a_{22} & a_{23} & a_{24}\\
            a_{31} & a_{32} & a_{33} & a_{34}\\
            a_{41} & a_{42} & a_{43} & a_{44}\\
        \end{pmatrix}\begin{pmatrix}
            0 & 0 & 0 & 0\\
            0 & 1 & 0 & 0\\
            0 & 0 & 1 & 0\\
            0 & 0 & 0 & 2
        \end{pmatrix} =
        \begin{pmatrix}
            0 & 0 & 0 & 0\\
            0 & 1 & 0 & 0\\
            0 & 0 & 1 & 0\\
            0 & 0 & 0 & 2
        \end{pmatrix}\begin{pmatrix}
            a_{11} & a_{12} & a_{13} & a_{14}\\
            a_{21} & a_{22} & a_{23} & a_{24}\\
            a_{31} & a_{32} & a_{33} & a_{34}\\
            a_{41} & a_{42} & a_{43} & a_{44}\\
        \end{pmatrix}.
    \end{equation}

    \noindent This results in the following conditions:

    \begin{equation}
        \begin{pmatrix}
            0 & a_{12} & a_{13} & 2a_{14}\\
            0 & a_{22} & a_{23} & 2a_{24}\\
            0 & a_{32} & a_{33} & 2a_{34}\\
            0 & a_{42} & a_{43} & 2a_{44}\\
        \end{pmatrix} = 
        \begin{pmatrix}
            0 & 0 & 0 & 0\\
            a_{21} & a_{22} & a_{23} & a_{24}\\
            a_{31} & a_{32} & a_{33} & a_{34}\\
            2a_{41} & 2a_{42} & 2a_{43} & 2a_{44}\\
        \end{pmatrix},
    \end{equation}

    \noindent which implies in turn that

    \begin{equation}
        A = \begin{pmatrix}
            a_{11} & 0 & 0 & 0\\
            0 & a_{22} & a_{23} & 0 \\
            0 & a_{32} & a_{33} & 0 \\
            0 & 0 & 0 & a_{44}
        \end{pmatrix}.
    \end{equation}

    \noindent Since we are interested in the operators $A\in \SU(2^{2})$ we impose that $AA^\dagger = \mathbbm{1}$ and $\det\{A\}=1$.

    \begin{equation}
        \begin{pmatrix}
            a_{11} & 0 & 0 & 0\\
            0 & a_{22} & a_{23} & 0 \\
            0 & a_{32} & a_{33} & 0 \\
            0 & 0 & 0 & a_{44}
        \end{pmatrix} \begin{pmatrix}
            a_{11}^* & 0 & 0 & 0\\
            0 & a_{22}^* & a_{32}^* & 0 \\
            0 & a_{23}^* & a_{33}^* & 0 \\
            0 & 0 & 0 & a_{44}^*
        \end{pmatrix} = \begin{pmatrix}
            1 & 0 & 0 & 0\\
            0 & 1 & 0 & 0\\
            0 & 0 & 1 & 0\\
            0 & 0 & 0 & 1
        \end{pmatrix}.
    \end{equation}

    \noindent This implies that:

    \begin{equation}
        a_{11} = e^{i\theta_1}, \quad A_2A_2^\dagger:=\begin{pmatrix} a_{22} & a_{23} \\ a_{32} & a_{33} \end{pmatrix}\begin{pmatrix} a_{22}^* & a_{32}^*\\  a_{23}^* & a_{33}^*\end{pmatrix} = \begin{pmatrix} 1 & 0\\ 0 & 1\end{pmatrix}, \quad a_{44} = e^{i\theta_2},
    \end{equation}

    \noindent and since $\det\{A\}= 1 = e^{i(\theta_1+\theta_2)}\det\{A_2\}$ we need $\det\{A_2\}= e^{-i(\theta_1+\theta_2)}$. We can thus rewrite $A$ as:

    \begin{equation}
        A = \left(\begin{array}{ c | c | c}
                    e^{i\theta_1} &  &  \\
                    \hline
                    & e^{-i\frac{1}{2}(\theta_1+\theta_2)} U & \\
                    \hline
                    & & e^{i\theta_2}
            \end{array}\right) \quad \text{with} \quad U\in \SU(2).
    \end{equation}

    \noindent With this decomposition we obtain that $H^{(2)}\SU(4)$ is described by five parameters.
\end{example}

\subsection{Numerical evaluation of the number of dimensions}\label{ssec:simulation}

The operator space can be covered approximately by a parameterized quantum circuit. Unless the quantum circuit is made to exactly cover the whole operator space, it will only cover a connected patch thereof. An approximation of the space can be made by concatenating a series of mutually non-abelian gates, termed layers. The layered approach of small gates can never cover the whole space \cite{kazi_universality_2024}, but it provides a sufficient approximation. A direct measurement of the dimensions of the subgroup covered by an ansatz is computationally infeasible for four or more qubits. It is however possible to get an approximation of the number of dimensions. According to  \cite{ragone_unified_2024}, the variance of a circuit relates to the size of the algebra as
\begin{equation}
    \operatorname{Var}_\theta(\ell) = \frac{\mathcal{P}(\rho)\mathcal{P}(O)}{\dim(\mathcal{M}\mathfrak{su})},
\end{equation}
where $\ell$ is the loss function and $\mathcal{P}(\rho)\mathcal{P}(O)$ measure entanglement and locality of the circuit. By keeping the studied problem and quantum circuit construction constant, we restrain the  relationship of the variance to the dimensionality of the Lie algebra as
\begin{equation}
    \operatorname{Var}_\theta(\ell) \sim 1/\dim(\mathcal{M}\mathfrak{su}).
\end{equation}

The number of independent dimensions within the circuit can be approximated by the variance of the derivative of the circuit \cite{ragone_unified_2024}. We have evaluated this variance in simulation for quantum circuits ranging from 4 to 16 qubits. The variance of the circuit can be approximated by choosing random parameters and evaluating the gradient at that point, then evaluating the variance of the samples  \cite{mcclean_barren_2018}. Our results show that the scaling of the symmetry-restricted subgroup and the variance of the corresponding circuit align, shown in figure \ref{fig:simulation scaling}. 

While the cyclic-invariant and standard approach quantum circuits do not completely approximate the scaling of the whole symmetry-restricted subgroup, they do grow exponentially. For the permutation-invariant quantum circuit we recover the expected polynomial scaling. The details of the simulation are given in appendix \ref{sec:simulation-details}, together with the construction of the quantum circuits. All circuits have approximately the same number of parameters. It is therefore the construction of the circuit and its confinement to a specific symmetry-restricted subgroup that restricts its scaling. This result confirms our theoretical finding above.

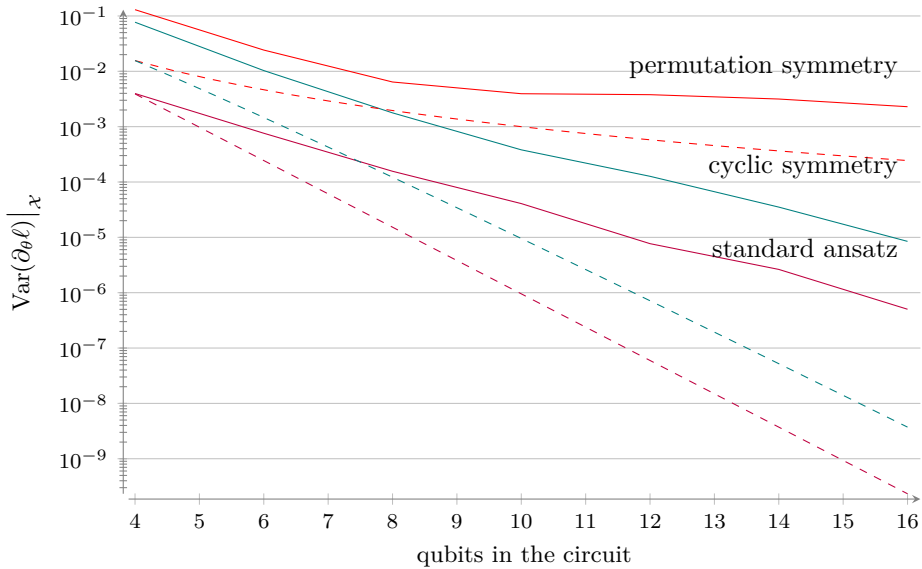
\begin{figure}
    \centering
\begin{tikzpicture}
\begin{axis}[small,
height=8cm, width=12cm,
no markers,
xlabel = {qubits in the circuit},
axis x line = bottom,
axis y line = left,
ymajorgrids,
major grid style = {very thin, gray!50},
major tick style = {very thin, gray!50},
%ytick align=outside,
axis line style={gray},
axis line shift=2pt,
ymode=log,
xmin = 3.9,
xmax = 16.2,
ylabel = {$\operatorname{Var}(\partial_\theta \ell) \big|_{\mathcal{X}}$},
%every axis y label/.style={at={(ticklabel cs:.5)},rotate = 90, anchor=center},
%ymax = 1.02,
%legend entries = {permutation symmetry, cyclic symmetry, cyclic symmetry 2, standard ansatz, strongly entangled standard ansatz},
%legend pos=south east,
]
%
%%% data needs to be in column format
% epch, Sn, Cn, Cn2, entanglement, strong entanglement
% 1, 0.1, 0.2, 0.3, 0.4, 0.5
% 2, 0.1, 0.2, 0.3, 0.4, 0.5
%
% and so on
\addplot [red] table [x=qubits, y=Sn_circuit, col sep=semicolon]{variance-loss-pivot.csv};
%\addplot [olive] table [x=qubits, y=Cn_circuit, col sep=semicolon]{variance-loss-pivot.csv};
\addplot [teal] table [x=qubits, y=Cn_circuit2, col sep=semicolon]{variance-loss-pivot.csv};
\addplot [purple] table [x=qubits, y=strongly_entanglement_circuit, col sep=semicolon]{variance-loss-pivot.csv};

\addplot[dashed, red, domain=4:16, samples=21] {1/(x^3)};
\addplot[dashed, teal, domain=4:16, samples=21] {1 * x/(4^x)};
\addplot[dashed, purple, domain=4:16, samples=21] {1 /(4^x)};
\draw (axis cs: 16, 5e-3) node[above left] {permutation symmetry}
(axis cs: 16, 8e-5) node[above left] {cyclic symmetry}
%(axis cs: 50, .53) node[above left] {strongly entangled}
(axis cs: 16, 3e-6) node[above left] {standard ansatz};
\end{axis}
%\draw[red] (0,0) rectangle (12, 8);
\end{tikzpicture}
    \caption{The variance of the derivative of the loss function on the quantum circuits constructed to confirm with a permutation-, cyclic- and no symmetry. The variance is evaluated at a set of parameters $\mathcal{X}$. The dashed lines of the same color indicate the maximal scaling of the variance in the complete symmetry-restricted subgroup.}
    \label{fig:simulation scaling}
\end{figure}

%\subsection{Finding the dimension of the subgroup}

%In the cases of a symmetry restriction that can be expressed through permutations, determining the dimension of the resulting restricted Lie group is a combinatorial exercise. 

%\subsection{Implications on barren plateaus}

%The scaling of the symmetry-invariant subgroups $\mathcal{M}SU$ has implications for the applicability of quantum computing to computations. As shown \cite{cerezo_does_2023, ragone_unified_2023, arrasmith_effect_2021, schatzki_theoretical_2024}, the scaling of the dimensions of the subgroup has implications on the appearance of barren plateaus in a variational quantum circuit setting.

\section{Overview of symmetry-restricted Lie groups}\label{ssec:summary}

The symmetry-restricted Lie groups have a smaller dimension that the non-restricted spaces. The number of dimensions of the group is equal to the number of independent elements in the algebra -- The symmetrized sums of Pauli strings. Counting the number of independent sums is  based on basic combinatorics. We show the order of growth of the dimensions of the Lie groups under the action of different symmetry groups in table \ref{tab:scalings}.

\begin{table}[h]
    \centering
    \caption{Overview the growth of dimensions of the subgroups for different symmetry restrictions. The symmetrized Lie groups are listed in order of the size of the corresponding discrete symmetry. For the energy invariance symmetry, there is no corresponding expression in the number of permutations.}
    \begin{tabular}{r | l l}
        Symmetry group & Growth of Symmetry & Scaling of dimensions \\\hline
        unconstrained Lie group & $\mathcal{O}(1)$ & $\mathcal{O}(4^N)$\\
        Cyclic group $C_N$ & $\mathcal{O}(N)$ & $\mathcal{O}(4^N/N)$\\
        Dihedral group $D_N$ & $\mathcal{O}(N)$ & $\mathcal{O}(4^N/N)$\\
        Alternating group $A_N$ & $\mathcal{O}(N!/2)$ & $\mathcal{O}(N^3)$\\
        Permutation group $S_N$ & $\mathcal{O}(N!)$ & $\mathcal{O}(N^3)$\\
        Permutation group in all entries $S_{2^N}$ & $\mathcal{O}(2^N!)$ & $\mathcal{O}(1)$\\\hline
        Energy conservation & & $\mathcal{O}(4^N/\sqrt{N})$\\\hline
    \end{tabular}
    
    \label{tab:scalings}
\end{table}

The relation between the size of the symmetry group acting on the Lie group and the dimension of the resulting restricted space is notable. In the case of the Alternating and Permutation group, the growth of the symmmetry group is faster than that of the unconstrained space, yet the dimension of the Lie group under the action of the symmetry is still growing, albeit polynominally.

%Some of these results, in particular the permutation invariance, also make an appearance in other fields of physics. The symmetry appears in a number of quantum mechanical systems, in particular the Bose gas \cite{ziff_ideal_1977}

%%%%%%%%%%%%%%%%%%%%%%%%%%%%%%%%%%%%%%%%%%%%%%%%%%%%%%%%%%%%

\section{Discussion}\label{sec:discussion}

We formalize the structure of subgroups of the special unitary group $\SU(2^N)$. Symmetry-restricted subgroups and the quantum circuits that implement them are well-behaved and scale depending on the structure and size of the restricting symmetry. We  confirm numerically that the restriction on the subgroup in turn restricts the expressibility of the quantum circuit, if it is constructed to be confined to that domain. Our results show that the scaling of the symmetry-restricted subgroup can be obtained by their relation to combinatorics for discrete symmetries and with a change of basis for matrix-based methods. For these two approaches, we also show that the structures are well-behaved. Using this result, we can directly employ the symmetry-restricted subspace, rather than skirting around the topic.

Since the scaling of the subgroup has a direct connection to barren plateaus \cite{ragone_unified_2024}, calculating the scaling for a given symmetry enables a performance estimation for training the quantum circuit. Even simple quantum circuit ansatze mirror the behaviour of their surrounding subgroup in terms of scaling, with a clear difference in the polynomial scaling of the permutation-invariant ansatz and the exponential scaling of cyclic invariance or the unconstrained standard ansatze.

Understanding the scaling of symmetry-restricted quantum circuits directly and their mathematical properties opens up avenues for direct treatment and exploitation of these structures. The fact that symmetry-restricted subgroups inherit much their supergroup's properties for wide case of applications puts prior work on sure footing and enables a more targeted treatment of phenomena like barren plateaus, but also the observed performance advantage of restricted quantum circuits.

Future work on this topic includes exploring continuous symmetries as restrictions. While quantum circuits can be constructed to include a continuous $\SU(2)$ \cite{nguyen_theory_2024}, they do not fit into the mathematical framework presented in this work. It can be expected that not all continuous symmetries can be included as restrictions, as a counterexample can be found easily, with the set $\mathcal{M}=\SU(2^N)$, because $\mathcal{M}\SU(2^N)$ equals the center of $\SU(2^N)$ \cite{lechner_private_2024}.

\section*{Acknowledgements}
Removed for peer review
%    The authors acknowledge funding from the German Federal Ministry of Education and Research (BMBF) under the funding program "Förderprogramm Quantentechnologien – von den Grundlagen zum Markt" (funding program quantum technologies – from basic research to market), project BAIQO, 13N16089.

%%%%%%%%%%%%%%%%%%%%%%%%%%%%%%%%%%%%%%%%%%%%%%%%%%%%%%%%%%%%
%%%%%%   Appendices  %%%%%%%

%\bibliography{references, zotero-library}

\appendix
\section{Lexicographic order}\label{sec:lexicographic-order}
    
    \begin{definition}
        Let $V,W$ be $\mathbb{R}-$vector spaces with bases $\{v_i\}_{i=0}^{N-1}$ and $\{w_i\}_{i=0}^{M-1}$. The \textbf{tensor product vector space} $V\otimes W$ is defined as the $\mathbb{R}-$vector space with basis given by the symbols $\{v_i\otimes w_j\}_{i,j=0}^{N-1, M-1}$.
    \end{definition}

    With $V,W$ defined as above, since $V\otimes W$ is a real vector space of dimension $N \times M$, we may consider the canonical basis $\left\{e_{k}\right\}_{k=0}^{N \times M-1}$ where $e_{k}$ are the usual column vectors having a $1$ at row $k$ and zero elsewhere. There is a \emph{natural} way to pair $v_{i} \otimes w_{j}$ to $e_{k}$, namely the so-called \textbf{lexicographic order}:

    \begin{equation}
        \begin{array}{llll}
        v_{0} \otimes w_{0} \leftrightarrow e_{0} & v_{1} \otimes w_{0} \leftrightarrow e_{M} & \cdots & v_{N} \otimes w_{0} \leftrightarrow e_{(N-1)M} \\
        v_{0} \otimes w_{1} \leftrightarrow e_{1} & v_{1} \otimes w_{1} \leftrightarrow e_{M+1} & & v_{N-1} \otimes w_{1} \leftrightarrow e_{(N-1)M+1} \\
        \vdots & \vdots & \ddots & \vdots \\
        v_{0} \otimes w_{M-1} \leftrightarrow e_{M-1} & v_{1} \otimes w_{M-1} \leftrightarrow e_{2M-1} & \cdots & v_{N-1} \otimes w_{M-1} \leftrightarrow e_{NM-1}
        \end{array}
    \end{equation}
        
    The interesting aspect of using the lexicographic order convention to identify the abstract tensor product construction with the canonical basis is that it agrees with the convention of identifying the tensor product with the Kronecker product. We can see this in the following example.
    
    \begin{example}
        Consider the real vector spaces $\mathbb{R}^{2}$ and $\mathbb{R}^3$ with the canonical bases given respectively by:
        \begin{equation}
            \left\{v_{i}\right\}_{i=0}^{1}=\left\{\left(\begin{array}{l}1 \\ 0\end{array}\right),\left(\begin{array}{l}0 \\ 1\end{array}\right)\right\} \quad \text{and} \quad \left\{w_{j}\right\}_{j=0}^{2}=\left\{\left(\begin{array}{l}1 \\ 0 \\ 0\end{array}\right),\left(\begin{array}{l}0 \\ 1 \\ 0\end{array}\right),\left(\begin{array}{l}0 \\ 0 \\ 1\end{array}\right)\right\}.
        \end{equation}
    
        \noindent If one constructs the tensor product vector space $\mathbb{R}^2\otimes \mathbb{R}^3$ with basis given by $\{v_i\otimes w_j\}_{i,j=0}^{1,2}$, one can show that the pairing defined above agrees with the Kronecker product $\otimes_K$ of the vectors. Indeed,

        \begin{align}
            & \begin{pmatrix}
            1 \\
            0
            \end{pmatrix} \otimes
            \begin{pmatrix}
            1 \\
            0 \\
            0
            \end{pmatrix} \leftrightarrow
            \begin{pmatrix}
            1 \\
            0 \\
            0 \\
            0 \\
            0 \\
            0
            \end{pmatrix} = 
            \begin{pmatrix}
            1 & \cdot\begin{pmatrix}
                1 \\
                0 \\
                0 \\
                \end{pmatrix}\\
            0 & \cdot\begin{pmatrix}
                1 \\
                0 \\
                0 \\
                \end{pmatrix}
            \end{pmatrix} = 
            \begin{pmatrix}
            1 \\
            0
            \end{pmatrix} \otimes_K
            \begin{pmatrix}
            1 \\
            0 \\
            0
            \end{pmatrix},
            \\
            & \begin{pmatrix}
            1 \\
            0
            \end{pmatrix} \otimes
            \begin{pmatrix}
            0 \\
            1 \\
            0
            \end{pmatrix} \leftrightarrow
            \begin{pmatrix}
            0 \\
            1 \\
            0 \\
            0 \\
            0 \\
            0
            \end{pmatrix} = 
            \begin{pmatrix}
            1 & \cdot\begin{pmatrix}
                0 \\
                1 \\
                0 \\
                \end{pmatrix}\\
            0 & \cdot\begin{pmatrix}
                0 \\
                1 \\
                0 \\
                \end{pmatrix}
            \end{pmatrix} = 
            \begin{pmatrix}
            1 \\
            0
            \end{pmatrix} \otimes_K
            \begin{pmatrix}
            0 \\
            1 \\
            0
            \end{pmatrix} ,
        \end{align}
    
        \noindent and the pairing continues in this fashion. We observe that the lexicographic identification matches the result of the Kronecker product of the vectors. This is, of course, a manifestation of the fact that the choice of the lexicographic order is equivalent to the choice of a basis in which the tensor product is represented by the Kronecker product of vectors.

    \end{example} 
    Having defined these concepts we proceed to construct a basis for the Lie algebra $\su(2^N)$.
    
    \begin{theorem}
        Let $1\leq N \in \mathbb{N}$. The ordered set

        \begin{equation}
            \mathcal{B}:= \left\{ \left\{i(\sigma_{\mu_1} \otimes \sigma_{\mu_2} \otimes \hdots \otimes \sigma_{\mu_N}) \right\}_{\mu_1 , \hdots , \mu_N=0}^3 \setminus \left\{i(\sigma_0 \otimes \hdots \otimes \sigma_0) \right\}   \right\}
        \end{equation}

        \noindent defines, upon the choice of the lexicographic ordering, a basis for the real vector space $\mathfrak{su}(2^{N})$.
    \end{theorem}
    \begin{proof}
        Since $\{\sigma_{i}\}_{i=0}^{3}$ are $2\times 2$ matrices, we may use the lexicographic ordering convention and take $\otimes = \otimes_K$, so that the elements in $\mathcal{B}$ can be regarded as matrices\footnote{This also explains why $i(\sigma_0\otimes\cdots\otimes\sigma_0)$ has to be removed, because once it is seen as a matrix, it is linearly dependent on the rest.} of size $2^{N}\times 2^{N}$.
        
        We note that there are a total of $4^{N}-1$ elements in $\mathcal{B}$ that are linearly independent, thanks to the tensor product construction. These constitute a basis of an $\mathbb{R}-$vector space of the same dimension as $\mathfrak{su}(2^{N})$. They also satisfy

    \begin{equation}
        \left(i\left(\sigma_{\mu_{1}} \otimes \ldots \otimes \sigma_{\mu_{N}}\right)\right)^{\dagger}=-i\left(\sigma_{\mu_{1}}^{\dagger} \otimes \cdots\otimes \sigma_{\mu_{N}}^{\dagger}\right)=-i\left(\sigma_{\mu_{1}} \otimes \cdots\otimes \sigma_{\mu_N}\right),
    \end{equation}
    
    \noindent i.e. they are skew-hermitian. Moreover, since $\operatorname{Tr}\{A \otimes B\}=\operatorname{Tr}(A) \cdot \operatorname{Tr}(B)$, all matrices have zero trace. This agrees with all the properties required in the definition \ref{def:su(N)} and thus, they form a basis of $\mathfrak{su}(2^N)$.
    \end{proof}

\section{Proof of Theorem \ref{thm:connectivitySU}}
\label{sec:proofconnected}

\subsection{Logical Steps of the proof}
\begin{itemize}
    \item Any $A \in \SU(n)$ can be diagonalized (Lemma \ref{def:diagonalbasic}). We define a diagonal matrix $D_A(t)$ in \ref{def:diagonalparam} which will be used later.
    \item There is a path in $\SU(n)$ connecting any $A \in \SU(n)$ to an element in the center $Z(\SU(n))$ (Lemma \ref{lemma:fifthlemma}).
    \item In the \ref{lemma:sixthlemma} we prove that there is a path in $\mathcal{M}\SU(n)$ which connects any $A \in \mathcal{M}\SU(n)$ to some element of the center of $\SU(n)$.
    \item In \ref{lemma:centers} we show that there are paths in $\SU(n)$ connecting any two elements of $Z(\SU(n))$.
    \item The key to proving the connectedness of $\mathcal{M}\SU(n)$ is to determine whether these paths connecting elements of $Z(\SU(n))$ fall entirely within $\mathcal{M}\SU(n)$. This depends on the properties of $\mathcal{M}$. When $\mathcal{M}$ is a set of commuting and diagonalizable matrices or $\mathcal{M}$ arises from the representation of a permutation group, then the elements of $Z(\SU(n))$ are path connected in $\mathcal{M}\SU(n)$, implying $\mathcal{M}\SU(n)$ is connected.
\end{itemize}

In this appendix we include the details of the proof of theorem \ref{thm:connectivitySU}. We start by providing some lemmas and definitions which will be used for the proof.

\subsection{The Proof}
\begin{lemma}
\label{def:diagonalbasic}
 Any $A\in SU(n)$ can be diagonalized as 
 \begin{equation}
     A=PD_AP^{\dagger},   
 \end{equation}
 where $D_A=\operatorname{diag}(e^{i\theta_1},e^{i\theta_2},\dots,e^{i\theta_{n}})$ and $\theta_j\in[0,2\pi)$ for every $j\in \{1,\ldots, n\}$.
\end{lemma}
\begin{proof}
All unitary matrices are normal matrices and all normal matrices are unitarily diagonalizable, see section 2.5 in \cite{Horn_Johnson_1985} for full proof.
\end{proof}
\begin{remark}
     In the remainder of this section, we shall use the notation $D_A$ and $P$ to refer to the matrices that diagonalize the matrix $A$.
\end{remark}

\begin{definition}
\label{def:diagonalparam}
    Given a matrix $A\in \SU(n)$, with the decomposition $A = PD_AP^\dagger$ as above, we define for each $t\in[0,1]$ the \textbf{parametrized diagonal matrix} $D_A(\cdot)\colon[0,1]\to \SU(n)$ as
    
\begin{equation}
    D_A(t)=\operatorname{diag}(e^{i(t\theta_1+(1-t)\varphi_k)},e^{i(t\theta_2+(1-t)\varphi_k)},\dots,e^{i(t\theta_{n}+(1-t)\varphi_k)}),
\end{equation}
where $\varphi_k := \frac{2\pi k}{n}$ for a fixed $k\in\{0,1,\ldots, n-1\}$, and from it define $A(t):=PD_A(t)P^\dagger$.
\end{definition}
\begin{remark}
    The reason to define $\varphi_k$ as above is that we will show that any element $A \in \SU(n)$ for which the sum of the exponents of the eigenvalues add up to $2\pi k$ for some $k\in\{0,1,\ldots, n-1\}$ (we will say such a matrix $A$ is of \textbf{type $k$}) is connected to the element of $Z\left(\SU(2^N)\right)$, of the form $\zeta_k\mathbbm{1}_{n}:=e^{i\frac{2k\pi}{n}}\mathbbm{1}_{n}$.

    We recall that the center of $\SU(n)$ consists of the matrices $\zeta_k\mathbbm{1}_{n}$ for $k\in\{0,1,\ldots,n-1 \}$ where $\zeta_k$ is the root of unity we just defined.
\end{remark}

\begin{lemma}
\label{lemma:fifthlemma}
    Let $A,P,D_A, D_A(t)$ and $\theta_i$ for $i\in \{1,\ldots, n\}$ be as defined in \ref{def:diagonalbasic} and \ref{def:diagonalparam}. If $A$ is of type $k$, namely, $\sum_{i=1}^n\theta_i=2\pi k$ for a fixed $k\in\{0,1,\dots,n-1\}$, then there is a path in $\SU(n)$ connecting $A$ and $\zeta_k\mathbbm{1}_{n}\in Z\left(\SU(n)\right)$.
\end{lemma}

\begin{proof}
 We recall here that $A(t):=PD_A(t)P^{\dagger}$ for $t\in[0,1]$, where

\begin{equation}
    D_A(t)=\operatorname{diag}(e^{i(t\theta_1+(1-t)\varphi_k)},e^{i(t\theta_2+(1-t)\varphi_k)},\dots,e^{i(t\theta_{n}+(1-t)\varphi_k)}).
\end{equation}

\noindent Then $A(t)A(t)^{\dagger}=A(t)^{\dagger}A(t)=\mathbbm{1}_{n}$ and

\begin{align}
    \det\{A(t)\}&=\det\{D(t)\}=\exp\left(i\sum_{i=1}^{n}(t\theta_i+(1-t)\varphi_k)\right)\\
    &=\exp(i(t\cdot2k\pi+(1-t)\cdot2k\pi))=\exp(i2k\pi)=1.
\end{align}

\noindent Therefore, $A(t)$ is in $\SU(n)$ for all $t\in[0,1]$ and is thus a path connecting $A$ and $\zeta_k\mathbbm{1}_{n}$.
\end{proof}

\begin{lemma}
\label{lemma:sixthlemma}
Let $k\in\{0,1,\ldots, n-1\}$ be fixed. If $A\in  \mathcal{M}\SU(n)$ for any $ \mathcal{M}\subseteq M_{n}(\mathbb{C})$ is of type $k$, then $A(t)\in  \mathcal{M}\SU(n)$ is a path connecting $A$ and $\zeta_k\mathbbm{1}_{2^N}$.
\end{lemma}

\begin{proof}

Considering the eigenvalues of $A$ may have multiplicities larger than one, we let $\mathcal{S}_1,\ldots, \mathcal{S}_m$, where $m\leq n$ denote the eigenspaces of $A$ corresponding to the  distinct eigenvalues $e^{i\theta_1},\ldots, e^{i\theta_m}$ and $\dim\{\mathcal{S}_l\} = N_l$. The diagonalization of $A$ reads $A = PD_AP^\dagger$, where

\begin{equation}
    D_A=\operatorname{diag}(\underbrace{e^{i\theta_1},\dots,e^{i\theta_1}}_{N_1},\dots,\underbrace{e^{i\theta_m},\dots,e^{i\theta_m}}_{N_m}).
\end{equation}

\noindent Since $AP=PD_A$, the $i-$th column $P^{(i)}$ of $P$, which is nothing but an eigenvector of $A$ corresponding to the eigenspace $\mathcal{S}_l$ for some $1\leq l \leq m$, satisfies $AP^{(i)}=e^{i\theta_l}P^{(i)}$. Moreover, for any element $S\in  \mathcal{M}$ it follows that

\begin{equation}
    AS=SA\Leftrightarrow AS = SPD_AP^{\dagger}\Leftrightarrow  A(SP) = (SP)D_A \Rightarrow A(SP)^{(i)}=e^{i\theta_l}(SP)^{(i)}.
\end{equation}

\noindent Since $A(SP)^{(i)}=e^{i\theta_l}(SP)^{(i)}$ for all $i\in\{1,\ldots, n\}$, then for each $1\leq i\leq n$ there exists some $1\leq l \leq m$ such that $(SP)^{(i)}\in\mathcal{S}_l$. Thus, we may rewrite $(SP)^{(i)}$ in terms of the eigenvectors of $\mathcal{S}_l$:

\begin{equation}
    (SP)^{(i)}=\sum_{P^{(s)}\in \mathcal{S}_l}c_sP^{(s)} \quad c_s\in\mathbb{C}.
\end{equation}

\noindent For $t\in[0,1]$, by the construction of $D_A(t)$ and $A(t)$ we also have $A(t)P^{(i)}=e^{i(t\theta_l+(1-t)\varphi_k)}P^{(i)}$. Note, moreover, that if $P^{(i)}, P^{(j)}\in\mathcal{S}_l$, then $A(t)P^{(i)}$ and $A(t)P^{(j)}$ share the same eigenvalue\footnote{This is not trivial, because if $A(t) = \text{diag}(e^{it\pi}, e^{-it\pi})$, for $t=1$, the eigenspace has dimension two, and the degeneracy is broken for $0<t<1$. However, thanks to our decomposition, this will not happen.} $e^{i(t\theta_l+(1-t)\varphi_k)}$.  Indeed, suppose this is not the case. Then, there exists some $\delta\in[0,2\pi)$ such that

\begin{equation}
    e^{i(t\theta_l+(1-t)\varphi_k) + \delta} =  e^{i(t\theta_l+(1-t)\varphi_k)}
\end{equation}

from which $e^{\delta}=1$, i.e. $\delta = 2k\pi$ for some $k\in\mathbb{Z}$, but the only possibility is $\delta = 0$, i.e. they are equal.

It will also not happen that two eigenvectors $P^{(i)}\in\mathcal{S}_l, P^{(j)}\in\mathcal{S}_m$ in different eigenspaces will share the same $A(t)-$eigenspace. Indeed,

\begin{equation}
    e^{i(t\theta_l+(1-t)\varphi_k)} = e^{i(t\theta_m+(1-t)\varphi_k)} \Leftrightarrow e^{it(\theta_l-\theta_m)}=1 \Leftrightarrow it(\theta_l-\theta_m) = 2\pi k \text{ for } k\in\mathbb{Z}.
\end{equation}

\noindent However, since $0\leq\theta_l,\theta_m<2\pi$, then $-2\pi<|\theta_l-\theta_m|<2\pi$ and since $t\in[0,1]$, we would only have $\theta_l = \theta_m$ as the only reasonable possibility. This contradicts our assumption that $P^{(i)}\in\mathcal{S}_l, P^{(j)}\in\mathcal{S}_m$ for $m\neq l$.

\noindent Hence, the eigenspace corresponding to the eigenvalue $e^{i\theta_l}$ of $A$ is the same as that corresponding to the eigenvalue $e^{i(t\theta_l+(1-t)\varphi_k)}$ of $A(t)$. Thus,

\begin{equation}
    A(t)(SP)^{(i)}=\sum_{P^{(s)}\in \mathcal{S}_l}c_sA(t)P^{(s)}=e^{i(t\theta_l+(1-t)\varphi_k)}\sum_{P^{(s)}\in \mathcal{S}_l}c_sP^{(s)}=e^{i(t\theta_l+(1-t)\varphi_k)}(SP)^{(i)},
\end{equation}

\noindent which implies

\begin{equation}
    A(t)SP=SPD(t)\Leftrightarrow A(t)S=SPD(t)P^{-1}=SA(t).
\end{equation}

\noindent We have thus shown that for every $t\in[0,1]$, $A(t)S = SA(t)$, i.e. $A(t)\in\mathcal{M}\SU(n)$. Moreover, since $A(0)=e^{i\varphi_k}\mathbbm{1}_{n}$ and $A(1)=A$, it follows that $A(t)$ is a path connecting $\zeta_k\mathbbm{1}_n$ and $A$ in $\mathcal{M}\SU(n)$.
\end{proof}

Now, we show that there is a path in $\SU(n)$ connecting any two elements in $Z(\SU(n))$.

\begin{lemma}
\label{lemma:centers}
For $k\neq l\in\{0,1,\ldots, n-1\}$, there are paths in $\SU(n)$ connecting $\zeta_k\mathbbm{1}_{n}$ and $\zeta_l\mathbbm{1}_{n}$, i.e. the elements of the center $Z\left(\SU(n)\right)$ are connected in $\SU(n)$.
\end{lemma}

\begin{proof}
Without loss of generality we may prove that for any $k$, there exists a path connecting $\mathbbm{1}_{n}$ with $\zeta_k\mathbbm{1}_{n}$. Notice that

\begin{equation}
    \zeta_k\mathbbm{1}_{n}=\operatorname{diag}(e^{i\frac{2k\pi}{n}},\dots,e^{i\frac{2k\pi}{n}},e^{i(\frac{2k\pi}{n}-2k\pi)},e^{i\frac{2k\pi}{n}},\dots,e^{i\frac{2k\pi}{n}}),
\end{equation}

\noindent where $e^{i(\frac{2k\pi}{n}-2k\pi)}$ is at the $h-$th position for any $h=1,\dots,n$. Define the path

\begin{equation}
    \Gamma_h(t)=\operatorname{diag}(e^{it\frac{2k\pi}{n}},\dots,e^{it\frac{2k\pi}{n}},e^{it(\frac{2k\pi}{n}-2k\pi)},e^{it\frac{2k\pi}{n}},\dots,e^{it\frac{2k\pi}{n}}).
\end{equation}

\noindent It is easy to see that all paths $\Gamma_h(t)$ of the form above are in $\SU(n)$ and verify $\Gamma_h(1)=\zeta_k\mathbbm{1}_n$ and $\Gamma_h(0)=\mathbbm{1}_n$.
\end{proof}

Equipped with these lemmas, we are finally in a position to prove theorem \ref{thm:connectivitySU}, in this case, setting $n=2^N$.

\begin{proof}{Proof of theorem \ref{thm:connectivitySU}.}
    \begin{enumerate}
        \item Let $\mathcal{M}:=\{C_i\}_{i\in I}$ be a subset of matrices of $M_{2^N}(\mathbb{C})$ that commute and are diagonalizable, where $I$ is some set of indices, not necessarily finite. From the definition of $\mathcal{M}\SU(2^N)$, see equation \eqref{eq:definvariantSU}, we observe that given a matrix $A\in \mathcal{M}\SU(2^N)$, the set $\{A\}\cup \mathcal{M}$ consists of commuting diagonalizable matrices. It is a known fact that they are all simultaneously diagonalizable over $\mathbb{C}$, i.e. there exists a matrix $P$ such that $PAP^{-1}=D_A$ and $PC_iP^{-1}=D_i$ ($D_A, D_i$ are diagonal matrices) for all $i\in I$. Since $A$ is in $\SU(2^N)$, we have the decomposition $A = PD_AP^\dagger$.

        %\begin{equation}
            %D_A = \operatorname{diag}\{\lambda_1,\ldots, \lambda_{2^N}\}, \quad |\lambda_i|^2 = 1\,\forall i=1,\ldots, 2^{N},\quad \prod_{i=1}^{2^N}\lambda_i = 1,
        %\end{equation}

        %where there might be repeated $\lambda$'s. The above conditions reduce to

        %\begin{equation}
            %D_A = \operatorname{diag}(e^{i\theta_1}, e^{i\theta_2},\ldots, e^{i\theta_{2^{N}-1}}, e^{i\theta_{2^{N}}}), \quad \sum_{i=1}^{2^{N}}\phi_i = 2\pi k,k\in \mathbb{N}\cup\{0\}.
        %\end{equation}

        Let us now define a path connecting $A$ to $\mathbbm{1}$ as follows:

        \begin{align}
            \Gamma\colon [0,1] &\to \SU(2^{N}) \\
                        t &\mapsto \Gamma(t) := P^{-1}D_A(t)P:=P^{-1}\operatorname{diag}(e^{it\theta_1}, e^{it\theta_2},\ldots, e^{it\theta_{2^{N}-1}}, e^{it\theta_{2^{N}}})P.
        \end{align}

        \noindent Note that $\Gamma(0)=P^{-1}P=\mathbbm{1}$ and that $\Gamma(1) = A$. Note, moreover, that $\Gamma(t)\in \mathcal{M}\SU(2^N)$ for all $t\in [0,1]$, since for any $i\in I$,

        \begin{equation}
            P[\Gamma(t), C_i] P^{-1} = P\Gamma(t)P^{-1}PC_iP^{-1} - PC_iP^{-1}P\Gamma(t)P^{-1}= [D_A(t), D_i]= 0
        \end{equation}

        Therefore, since we have found a path connecting an arbitrary matrix of $\mathcal{M}\SU(2^{N})$ with the identity that lies, for every $t\in [0,1]$ inside $\mathcal{M}\SU(2^{N})$, we are done. 

        % \item Since $\mathcal{M}$ is a set of commuting diagonalizable matrices, all elements in $G$ are simultaneously diagonalizable by some matrix $Q$. Now we claim that $QP_h(t)Q^{-1}\in G\SU(2^{N})$.

        % Indeed, for any $A\in G\SU(2^{N})$, $A$ can be diagonalized as $A=QD_A Q^{-1}$, then
        % $$(QP_h(t)Q^{-1})A=(QP_h(t)Q^{-1})(QD_A Q^{-1})=Q P_h(t)D_AQ^{-1}=Q D_AP_h(t)Q^{-1}=A(QP_h(t)Q^{-1}).$$
        % Hence the centers of $SU(2^{N})$ are path connected in $G\SU(2^{N})$ and by above lemmas we get that $G\SU(2^{N})$ is path connected.
        %\item (Proof of the first element of any $A\in\mathcal{M}$ is 1 if $\mathcal{M}$ is a SWAP group) We only need to show that the first elements of generators of SWAP are 1, because for permutation matrices obtained by multiplication and inversion of generators whose first elements are 1, their first elements must also be 1. 

        %Any permutation can be written as a product of adjacent transpositions and the adjacent transposition between the  $k$th and $(k+1)$th qubits ($k\in\{0,1,\dots, n-1\}$) can be written as
        %\begin{equation}
            %T_k = \underbrace{\mathbbm{1} \otimes \mathbbm{1} \otimes \ldots \otimes \mathbbm{1}}_{k-1} \otimes S \otimes \underbrace{\mathbbm{1} \otimes \mathbbm{1} \otimes \ldots \otimes \mathbbm{1}}_{n-k-1},
         %\end{equation}
        %where $S=\begin{pmatrix}
            %1 & 0 & 0 & 0 \\
            %0 & 0 & 1 & 0 \\
            %0 & 1 & 0 & 0 \\
            %0 & 0 & 0 & 1 \\
        %\end{pmatrix}$.
        %Notice that the first element of $T_k$ is the product of the first elements of the matrices on the right side, so we have $(T_k)_{11}=1$ for any $k\in\{0,1,\dots, n-1\}$.
        
        \item We already know that $A$ is connected to $\zeta_k\mathbbm{1}_{2^N}$ for some $k\in\{0,1,\ldots, 2^N-1\}$ through a path $\tilde{\Gamma}$. Moreover, the requirements on $\mathcal{M}$ we get that for any $A\in\mathcal{M}$, $A$ is a permutation matrix and $A_{11}=1$, as can be inferred from theorem \ref{th:Uimplementingsymmetry}. $A$ can be viewed as a block diagonal matrix, with an identity of size $1\times1$ in the upper left and a permutation matrix $\tilde{A}$ of size $(2^N-1)\times(2^N-1)$ in the lower right. Let $\Gamma(t)=\operatorname{diag}(e^{it(\varphi_k-2k\pi)},e^{it\varphi_k},\dots,e^{it\varphi_k})$ be the path in $\SU(2^N)$ connecting $\mathbbm{1}_{2^N}$ and $\zeta_k\mathbbm{1}_{2^N}$, see \ref{lemma:centers} for details. Then
        \begin{align}
            A\Gamma(t)A^{-1} &=\operatorname{diag}(e^{it(\varphi_k-2k\pi)},\tilde{A}e^{it\varphi_k}\mathbbm{1}_{2^N-1}\tilde{A}^{-1})\\
            & =\operatorname{diag}(e^{it(\varphi_k-2k\pi)},e^{it\varphi_k}\mathbbm{1}_{2^N-1})=\Gamma(t).
        \end{align}
        Therefore, $\Gamma(t)\in \mathcal{M}\SU(2^{N})$. Combining this path together with $\tilde{\Gamma}$ we conclude the proof.
    \end{enumerate}

\end{proof}

\section{Simulation details}\label{sec:simulation-details}

We document the details of the simulation performed in figure \ref{fig:simulation scaling}. We construct our circuits following \cite{larocca_group-invariant_2022} for permutation-invariant quantum circuits and generalize the structure to cyclic invariance. The  strong entangling ansatz is taken from \cite{schuld_circuit-centric_2020} and a standard building block in many quantum machine learning programming frameworks \cite{bergholm_pennylane_2022}.

The circuit structures used in these experiments are detailed in figure \ref{fig:ansatz-overview}. These quantum circuit layers are adapted to the respective number of qubits and are layered to cover a significant portion of the operator space. In the experiments, we aim to keep the number of parameters approximately the same across the different circuits. 

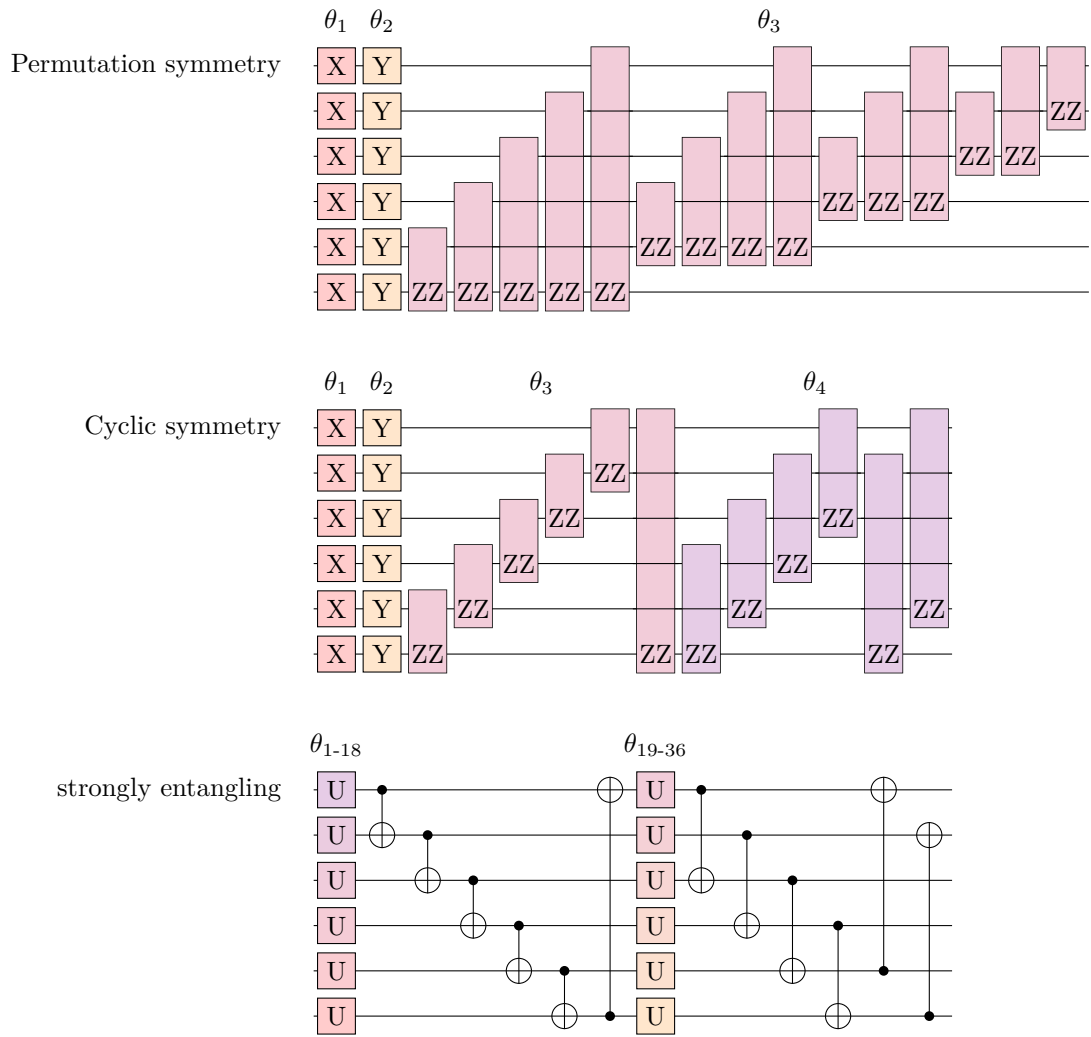
\begin{figure}
    \centering
      \begin{tikzpicture}[gate/.style={rectangle, draw=black, fill=white}, xscale=.6, yscale=.6]
    \begin{scope}[yscale=-1, yshift=-7cm]
\foreach \i/\col in {1/0, 2/20, 3/40, 4/60, 5/80, 6/100} {
\draw (0.5, \i) --+ (14, 0);
\draw (1, \i) node[gate, fill=red!\col!violet!20] {U};
	%(2, \i) node[gate, fill=orange!\col!purple!20] {Y};
	};
\foreach \i in {5,4, ..., 1} {
\path (\i + 1, \i+1) node[circle, draw=black] (target) {} (\i + 1, \i) node[circle, fill=black, inner sep=1.3pt] (control) {};
\draw (target.south) -- (control.center);};
\path (7, 1) node[circle, draw=black] (target) {} (7, 6) node[circle, fill=black, inner sep=1.3pt] (control) {};
\draw(target.north) -- (control.center);

\begin{scope}[xshift = 7cm]
\foreach \i/\col in {1/0, 2/20, 3/40, 4/60, 5/80, 6/100} {
\draw (1, \i) node[gate, fill=orange!\col!purple!20] {U};
	};
\foreach \i in {1,2, ..., 4} {
\path (\i + 1, \i+2) node[circle, draw=black] (target) {} (\i + 1, \i) node[circle, fill=black, inner sep=1.3pt] (control) {};
\draw (target.south) -- (control.center);};
\foreach \i in {1,2} {
\path (5 + \i, \i) node[circle, draw=black] (target) {} (5 + \i, 4 + \i) node[circle, fill=black, inner sep=1.3pt] (control) {};
\draw(target.north) -- (control.center);};
\end{scope}
\end{scope}
\draw (1, 7) node {$\theta_{1\text{-}18}$}  (8, 7) node {$\theta_{19\text{-}36}$};
\draw (0, 6) node[anchor=east] {strongly entangling};

\begin{scope}[yshift=8cm]
\foreach \i in {1,2,...,6} {
\draw (0.5, \i) --+ (14, 0);
\draw (1, \i) node[gate, fill=red!20] {X}
	(2, \i) node[gate, fill=orange!20] {Y};
	};
\foreach \i in {1,2, ..., 5} {
\draw[line width=.3pt, line cap=rect, double=purple!20, double distance=14pt] (\i + 2, \i) --+ (0, 1);
\draw (\i + 2, \i) node {ZZ};};
\draw[line width=.3pt, line cap=rect, double=purple!20, double distance=14pt] (8, 1) --+ (0, 5);
\draw (8,1) node {ZZ};
\foreach \i in {2, 3, ..., 5} {
\draw (7.5, \i) --+(1,0);};
\foreach \i in {1,2, ..., 4} {
\draw[line width=.3pt, line cap=rect, double=violet!20, double distance=14pt] (\i + 8, \i) --+ (0, 2);
\draw (\i + 7.5, \i+1) --+ (1,0);
\draw (\i + 8, \i) node {ZZ};};
\foreach \i in {1,2} {
\draw[line width=.3pt, line cap=rect, double=violet!20, double distance=14pt] (12 + \i, \i) --+ (0, 4);
	\foreach \j in {1, 2, 3} {
	\draw (11.5 + \i, \i + \j) --+ (1,0);};
\draw (12 + \i,\i) node {ZZ};};
\draw (1, 7) node {$\theta_1$} (2, 7) node {$\theta_2$} (5.5, 7) node {$\theta_3$} (11.5, 7) node {$\theta_4$};
\draw (0, 6) node[anchor=east] {Cyclic symmetry};
\end{scope}

\begin{scope}[yshift=16cm]
\foreach \i in {1,2,...,6} {
\draw (0.5, \i) --+ (17, 0);
\draw (1, \i) node[gate, fill=red!20] {X}
	(2, \i) node[gate, fill=orange!20] {Y};
	\foreach \j in {3, 4, ..., 8} {
	\node (p\j\i) at (\j, \i) {\phantom{$ZZ$}};};};
\foreach \i in {1,2, ..., 5} {
\draw[line width=.3pt, line cap=rect, double=purple!20, double distance=14pt] (\i + 2, 1) --+ (0, \i);
\draw (\i + 2, 1) node {ZZ};};

\foreach \i in {1,2, ..., 4} {
\draw[line width=.3pt, line cap=rect, double=purple!20, double distance=14pt] (\i + 7, 2) --+ (0, \i);
\draw (\i + 7, 2) node {ZZ};
\draw (\i + 2.5, \i + 1) --+ (5 - \i, 0);};

\foreach \i in {1,2, ..., 3} {
\draw[line width=.3pt, line cap=rect, double=purple!20, double distance=14pt] (\i + 11, 3) --+ (0, \i);
\draw (\i + 11, 3) node {ZZ};
\draw (\i + 7.5, \i + 2) --+ (4 - \i, 0);};

\foreach \i in {1,2} {
\draw[line width=.3pt, line cap=rect, double=purple!20, double distance=14pt] (\i + 14, 4) --+ (0, \i);
\draw (\i + 14, 4) node {ZZ};
\draw (\i + 11.5, \i + 3) --+ (3 - \i, 0);};

\foreach \i in {1} {
\draw[line width=.3pt, line cap=rect, double=purple!20, double distance=14pt] (\i + 16, 5) --+ (0, \i);
\draw (\i + 16, 5) node {ZZ};
\draw (\i + 14.5, \i + 4) --+ (2 - \i, 0);};
\draw (1, 7) node {$\theta_1$} (2, 7) node {$\theta_2$} (10.5, 7) node {$\theta_3$};
\draw (0, 6) node[anchor=east] {Permutation symmetry};
\end{scope}
  \end{tikzpicture}
    \caption{The different circuits used in the experiment with the assignment of parameters. The ansatz is repeated in layers until the number of parameters is the same across all ansatzes. Parameters are not shared between layers. Same color in the gates indicate shared parameters.}
    \label{fig:ansatz-overview}
\end{figure}
%\begin{figure}
%    \centering
%    \input{pi-circuit}
%    \caption{Structure of the permutation-invariant circuit used in the experiment shown in figure \ref{fig:simulation scaling}. The structure is repeated as a layer of the quantum machine learning ansatz. The parameters of each layer are independent.}
%    \label{fig:pi-circuit}
%\end{figure}

%\begin{figure}
%    \centering
%    \input{ci-circuit}
%    \caption{Structure of the cyclic-invariant quantum circuit used in the experiment of figure \ref{fig:simulation scaling}. The structure is repeated as a layer in the ansatz. Parameters are not shared between layers.}
%    \label{fig:ci-circuit}
%\end{figure}

%\begin{figure}
%    \centering
%    \input{se-circuit}
%    \caption{The standard strong entangling ansatz used for comparison. }
%    \label{fig:se-circuit}
%\end{figure}

Each $ZZ$ gate in the quantum circuit corresponds to a parameterized element of the algebra, $\exp(-i\theta/2 \sigma_z \otimes \sigma_z)$ distributed across the qubits. These gates can be implemented as quantum circuits as shown in figure \ref{fig:zz-gate}.

\begin{figure}
    \centering
    \begin{tikzpicture}
\foreach \i in {1,2} {
\draw (-.5,\i) --+ (3,0);};
\path (0,1) node[circle, fill=black, inner sep=2pt] {} (0,2) node[circle, draw=black] (c1){}
	(2,1) node[circle, fill=black, inner sep=2pt] {} (2,2) node[circle, draw=black] (c2){};
\draw (c1.north) -- (0,1) (c2.north) -- (2,1);
\draw (1, 2) node[rectangle, draw=black, fill=white] {$R_Z(\alpha)$};
\end{tikzpicture}

    \caption{An implementation of the $ZZ$ gate with elementary quantum circuit elements.}
    \label{fig:zz-gate}
\end{figure}
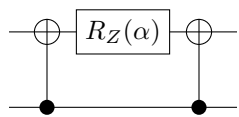
We compare the impact of the growth of two symmetry groups with different scalings in the context of supervised quantum machine learning.
To be precise, we have a dataset $\{(x_i, y_i)\}_{i=1}^M$ consisting of classical features $x_i$ and targets $y_i$.
The corresponding quantum input states $\rho_i$ are obtained via a feature map $f:x_i\mapsto\rho_i$.
A label is predicted by applying an ansatz $U(\boldsymbol\theta)$ to $\rho_i$ with trainable parameters $\boldsymbol\theta$ and estimating the expectation value $\ell_{\boldsymbol\theta}=\mathrm{Tr[U(\boldsymbol\theta)\rho_i U(\boldsymbol\theta)^\dagger O]}$ for an observable $O$.
The quality of the predictions are expressed through a loss function $\mathcal L(\boldsymbol\theta)=\frac{1}{N}\sum_{i=1}^Mf(\ell_{\boldsymbol\theta}(\rho_i),y_i)$.
For instance, choosing $f(\ell_{\boldsymbol\theta}(\rho_i), y_i)=(\ell_{\boldsymbol\theta}(\rho_i) - y_i)^2$ we obtain the mean squarer error $\mathcal L_{\mathrm{mse}}$.
More detailed information can be found in \cite{Larocca2022diagnosingbarren}.

In figure \ref{fig:simulation scaling} we compare the scaling of the variance of the partial derivatives of the mean squared error $\mathrm{Var}[\frac{\partial}{\partial\theta_i}\mathcal{L}_{\mathrm{mse}}(\boldsymbol\theta)]$, for increasing system sizes $n=4,6,\dots,16$.
Typically, gradient-based methods are used to find the minimizing direction for the weights.
Thus analysing how the derivatives concentrate provides insight into the trainability of the circuit \cite{Thanasilp2023}. The exact structure of the quantum circuit has a negligible impact on the quantum circuit. We have tried the numerical experiment with different variations of the cyclic invariant quantum circuit as well as different structures of the standard ansatz. 

\begin{figure}
    \centering
    \begin{tikzpicture}
\begin{axis}[small,
height=6cm, width=12cm,
no markers,
xlabel = {qubits in the circuit},
axis x line = bottom,
axis y line = left,
ymajorgrids,
major grid style = {very thin, gray!50},
major tick style = {very thin, gray!50},
%ytick align=outside,
axis line style={gray},
axis line shift=2pt,
ymode=log,
xmin = 3.9,
xmax = 16.1,
ylabel = {Variance of the derivative of the circuit},
%every axis y label/.style={at={(ticklabel cs:.5)},rotate = 90, anchor=center},
%ymax = 1.02,
%legend entries = {permutation symmetry, cyclic symmetry, cyclic symmetry 2, standard ansatz, strongly entangled standard ansatz},
%legend pos=south east,
]
%
%%% data needs to be in column format
% epch, Sn, Cn, Cn2, entanglement, strong entanglement
% 1, 0.1, 0.2, 0.3, 0.4, 0.5
% 2, 0.1, 0.2, 0.3, 0.4, 0.5
%
% and so on
\addplot [red, thick] table [x=qubits, y=Sn_circuit, col sep=semicolon]{variance-loss-pivot.csv};
\addplot [olive] table [x=qubits, y=Cn_circuit, col sep=semicolon]{variance-loss-pivot.csv};
\addplot [teal] table [x=qubits, y=Cn_circuit2, col sep=semicolon]{variance-loss-pivot.csv};
\addplot [purple] table [x=qubits, y=strongly_entanglement_circuit, col sep=semicolon]{variance-loss-pivot.csv};
\addplot [magenta] table [x=qubits, y=se2, col sep=semicolon]{variance-loss-pivot.csv};
\addplot [red!30!blue] table [x=qubits, y=simple_se, col sep=semicolon]{variance-loss-pivot.csv};

%\addplot[dashed, red, domain=3:16, samples=21] {2/(x^3)};
%\addplot[dashed, teal, domain=3:16, samples=21] {1 * x/(4^x)};
%\addplot[dashed, purple, domain=3:16, samples=21] {1 /(4^x)};
\draw (axis cs: 16, 3e-3) node[above left] {\small permutation symmetry}
(axis cs: 16, 6e-5) node[above left] {\small cyclic symmetry}
%(axis cs: 50, .53) node[above left] {strongly entangled}
(axis cs: 13, 1e-6) node[above left] {\small standard ansatz};
\end{axis}
%\draw[red] (0,0) rectangle (12, 8);
\end{tikzpicture}
    \caption{Extended numerical results for the variance simulation of figure \ref{fig:simulation scaling}. The different standard ansatzes refer to different $\CNOT$ patterns. For the two cyclic invariant curves, the difference is in whether the fourth parameter $\theta_4$ is present or not. The construction makes little difference, indicating that the dimension of the subspace dominates the result.}
    \label{fig:simulation-extended-results}
\end{figure}
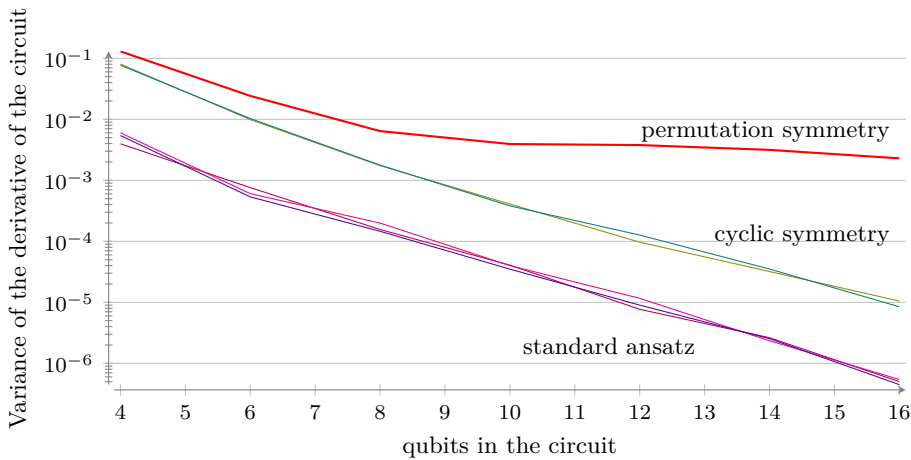

To this end, we considered the simple problem of classifying connected and disconnected graphs $G = (V, E)$ with $|V| = n$, similar to \cite{schatzki_theoretical_2024}. The classical graphs were embedded as quantum graph states. Specifically, the states \cite{Raussendorf2003}
\begin{equation}
    \ket{G}=\prod_{(a,b)\in E}CZ^{\{a, b\}}\ket{+}^{\otimes n}
\end{equation}
were prepared where $CZ^{\{a,b\}}$ is the controlled $Z$ gate acting on the qubits $a$ and $b$.
The figure \ref{fig:ansatz-overview} shows the construction of the permutation equivariant ansatz, cyclic equivariant ansatz and a standard ansatz \cite{schuld_circuit-centric_2020} as a base reference with polynomial and exponential scalings.
We measure the system with the permutation-invariant observable 

\begin{equation} 
Z^{\otimes n}=\sum_{i:\text{even}}\ket{i}\bra{i}-\sum_{i:\text{odd}}\ket{i}\bra{i}
\end{equation}

As a consistent estimator, we used the sample variance to estimate $\mathrm{Var}[\frac{\partial}{\partial\theta_i}\mathcal{L}_{\mathrm{mse}}(\boldsymbol\theta)]$.
We sampled $200$ parameter vectors uniformly from $[-2\pi,2\pi]$ for the respective circuit.
To evaluate the loss, a dataset of size $50$ with balanced classes was generated for each system size according to the Rènyi-Erd\H{o}s model with edge probability $0.4$.
The figure shows the scaling for a representative parameter in the middle of the circuit.
Similar scalings were observed for almost all parameters.

\end{document}